\documentclass[conference]{IEEEtran}
\IEEEoverridecommandlockouts

\usepackage{cite}
\usepackage{amsmath,amssymb,amsfonts}
\usepackage[ruled,linesnumbered]{algorithm2e} 
\usepackage{subcaption} 
\usepackage{graphicx}
\usepackage{textcomp}
\usepackage{xcolor}

\usepackage{amsthm}
\usepackage{url}

\usepackage{enumitem}

\usepackage{comment}

\newtheorem{theorem}{Theorem}
\newtheorem{lemma}{Lemma}
\newtheorem{definition}{Definition}
\newtheorem{example}{Example}
\newtheorem{remark}{Remark}

\newtheorem{prule}[remark]{Rule}

\begin{document}

\title{Lantern: Finding Committable Transactions via Back-Propagation on DAGs
\thanks{Tian Guan and Mingchao Wan are co-corresponding authors. This work was funded by the Beijing Academy of Blockchain and Edge Computing. Source code is available at \protect\url{https://github.com/lantern-li/Lantern}}
}

\author{
	\IEEEauthorblockN{
		Denglong Li\IEEEauthorrefmark{1}\IEEEauthorrefmark{2},
		Gerui Wang\IEEEauthorrefmark{2}, 
		Tian Guan\IEEEauthorrefmark{1}, 
		Mingchao Wan\IEEEauthorrefmark{2}
	}
	\IEEEauthorblockA{
		\IEEEauthorrefmark{1}\textit{Tsinghua University}, 
		\IEEEauthorrefmark{2}\textit{Beijing Academy of Blockchain and Edge Computing}
	}
	
	\IEEEauthorblockA{
		li-dl24@mails.tsinghua.edu.cn, wanggerui@baec.org.cn, guantian@sz.tsinghua.edu.cn, 
		chainmaker@baec.org.cn
	}
}
\maketitle

\begin{abstract}
Existing concurrency control protocols either introduce nondeterminism, resulting in a serial execution-replay dependency between primary and replica nodes, or rely on impractical prior knowledge of transaction read-write sets. 
In this paper, we present Lantern, a deterministic concurrency control protocol tailored for high-performance transaction processing systems operating without prior knowledge. 
The key insight of Lantern is that all zero-out-degree transaction vertices in the local dependency graph can be safely committed in ascending order using an overwrite-permissive strategy. 
We further introduce a novel Back-Propagation (BP) mechanism that iteratively propagates dependency states from sink to source vertices to identify additional committable transactions. 
We also propose Conflict-Free Batch Selection (CFBS) for read-modify-write intensive scenarios. 
We integrate Lantern into the open-source blockchain platform ChainMaker.
Extensive evaluations on YCSB and SmallBank benchmarks demonstrate that Lantern achieves up to a 5.9x throughput speedup over Aria and improves the throughput of ChainMaker's execution layer by at least 2.2x.
\end{abstract}


\section{Introduction}
\label{Introduction}

As modern transaction processing systems continue to scale, improving transaction throughput has become a critical challenge. 
While multi-core processors provide increasing computational resources for parallel transaction execution, exploiting intra-node parallelism in distributed systems remains challenging because all participating nodes must produce identical execution outcomes. 
This calls for deterministic concurrency control protocols that enable efficient parallel execution while ensuring serializability~\cite{ConcurrencyControlinDistributedDatabaseSystems} and consistency across nodes.

Despite their ability to exploit intra-node parallelism, many existing concurrency control protocols are nondeterministic~\cite{OCC,silo,Polaris,TicToc,Hekaton,Cicada}. 
Given the same initial state and transaction sequence, these protocols may produce different execution orders across nodes due to variations in runtime scheduling, potentially leading to divergent execution outcomes. 
To ensure consistency, such protocols typically adopt a \emph{two-stage execution architecture}~\cite{PDP,partdag,ParBlockchain,Eve,OptSmart,Dickerson,Dipetrans}. 
As illustrated in Fig.~\ref{towstageexecution}, the primary node first executes transactions in parallel using a nondeterministic protocol and generates \emph{scheduling metadata} that captures the dependencies established during execution. This metadata is then disseminated to replica nodes, which use it to replay the transactions and reproduce the same serializable outcome as the primary node.

This two-stage architecture introduces a strict execution dependency between the primary and replica nodes, limiting inter-node parallelism and increasing end-to-end execution latency. 
Specifically, it exhibits three key characteristics: (i) \emph{Cumulative Latency}, as transaction processing requires both the primary execution stage and the subsequent replica replay stage; (ii) \emph{Efficient Replay}, as conflict resolution has already been performed by the primary node, allowing replicas to focus on replaying transactions according to the scheduling metadata; and (iii) \emph{Metadata Dependency}, as replica execution cannot proceed until the primary node has generated and disseminated the required scheduling metadata. 
These characteristics make the two-stage architecture a potential bottleneck for high-throughput transaction processing. 
ChainMaker~\cite{chainmaker}, a widely deployed consortium blockchain platform in China~\cite{BeijingGov2026ChainMaker}, adopts an OCC-based~\cite{OCC} nondeterministic concurrency control protocol and relies on this two-stage architecture to ensure consistent execution across nodes. 
It therefore provides a representative system for evaluating the performance implications of this design.

In contrast, deterministic concurrency control protocols~\cite{Calvin,BOHM,PWV,QStore,Caracal,LiTM,DOCC,aria,blockstm,Spectrum} enable all participating nodes to execute transactions independently and in parallel while producing identical serializable outcomes. 
These protocols naturally fit the \emph{Order-Execute (OE) architecture}~\cite{FISCO-BCOS,HighthroughputByzantinefaulttolerance,QuorumGighub,blockstm}, in which nodes first establish a globally ordered transaction sequence through consensus and then independently execute the sequence using a deterministic protocol (as illustrated in Fig.~\ref{OE}). 
By eliminating the strict execution dependency between primary and replica nodes, the OE architecture avoids the cumulative latency inherent in two-stage execution. 
However, existing deterministic protocols differ in how they determine transaction dependencies. 
Static-analysis-based protocols~\cite{Calvin,BOHM,PWV,QStore,Caracal} require a priori knowledge of transaction read-write sets to construct an execution schedule before runtime. 
This requirement can be difficult to satisfy in blockchain systems because smart contracts are quasi-Turing-complete~\cite{EthereumYellowPaper}: their execution paths may depend on runtime data, and the state variables accessed may vary according to runtime conditions. 
Consequently, runtime deterministic protocols~\cite{LiTM,DOCC,aria,blockstm,Spectrum} that do not rely on prior knowledge of transaction read-write sets are essential for enabling practical parallel execution in such environments.

However, designing runtime deterministic concurrency control without prior knowledge of transaction read-write sets presents several challenges.
First, transaction read-write behavior can only be discovered during execution, making it difficult to identify and manage transaction dependencies without a priori information. 
Second, the protocol must guarantee not only serializability but also determinism: all participating nodes must produce the same serializable outcome under the same initial state and ordered transaction input, corresponding to a consistent serial execution order. 
Third, the combination of lacking prior knowledge of transaction read-write sets and the requirement for deterministic execution constrains scheduling and dependency resolution, making it challenging to achieve high throughput under high contention.

\begin{figure}[t]
	\centering
	\includegraphics[width=0.9\linewidth]{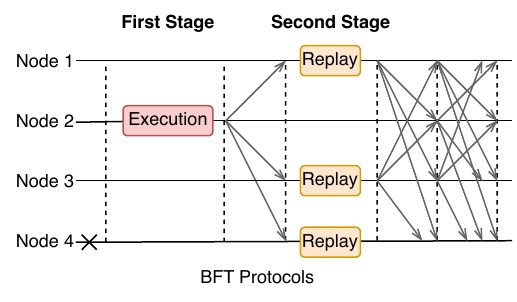}
	\caption{Two-Stage Execution Architecture}
	\label{towstageexecution}
\end{figure}

To address these challenges, we propose Lantern, a runtime deterministic concurrency control protocol that achieves high-throughput parallel transaction execution without requiring prior knowledge of transaction read-write sets. 
Lantern employs a batch-based graph construction scheme that confines dependency analysis to individual batches, enabling runtime dependency resolution. 
Within each batch, Lantern identifies all zero-out-degree vertices in the dependency graph as committable transactions and commits them in ascending transaction order, leveraging the overwrite semantics of serial execution to ensure deterministic outcomes.
To improve the transaction commit rate, Lantern further introduces a deterministic Back-Propagation (BP) mechanism that iteratively propagates dependency states from sink vertices to source vertices, for which we establish a rigorous liveness proof. 
Moreover, Lantern incorporates a Conflict-Free Batch Selection (CFBS) mechanism to proactively avoid conflicts in hot-account read-modify-write (RMW)-intensive workloads.

Beyond blockchain systems, Lantern is applicable to general-purpose database systems that provide a globally ordered transaction stream. In deterministic database architectures~\cite{AnEvaluationoftheAdvantagesandDisadvantagesofDeterministicDatabaseSystems}, a sequencing layer typically establishes a total order over incoming transactions using a replicated log maintained by a crash fault-tolerant (CFT) protocol. 
Given this globally agreed-upon transaction sequence, each node can independently execute transactions in parallel using Lantern, producing an identical serializable outcome across nodes.
This separation of transaction sequencing from execution allows Lantern to be integrated into diverse transaction processing systems, independently of the underlying consensus protocol. 

To demonstrate Lantern's practicality and performance, we have integrated it into the open-source blockchain platform ChainMaker~\cite{chainmaker} and conducted comprehensive experiments using the YCSB and SmallBank benchmarks. 
The results show that Lantern achieves up to $5.9\times$ the throughput of Aria~\cite{aria}. Furthermore, by eliminating the serial execution-replay dependency between primary and replica nodes, Lantern improves the throughput of ChainMaker's execution layer by at least $2.2\times$ compared with its original execution architecture. 

\begin{figure}[t]
	\centering
	\includegraphics[width=0.9\linewidth]{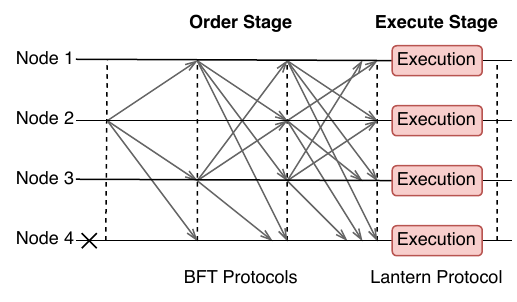}
	\caption{Order-Execute (OE) architecture}
	\label{OE}
\end{figure}

\section{Background}
This section first establishes the formal definitions of deterministic concurrency control, followed by the directed graph preliminaries essential to Lantern.

\subsection{Deterministic Concurrency Control}
Deterministic concurrency control (DCC) enables all nodes to independently execute the same ordered transaction sequence in parallel while guaranteeing identical serializable outcomes.
Unlike traditional nondeterministic protocols that only ensure serializability, DCC guarantees both serializability and determinism.
Let $S$ denote the initial world state, and let $T = \langle \mathit{TX}_0, \mathit{TX}_1, \dots, \mathit{TX}_{n-1} \rangle$ be a transaction sequence with a predefined order. 
The concepts of serializability and determinism are formally defined as follows:

\begin{definition}[Serializability]
	\label{def:serializability}
	A concurrent execution of $T$ is serializable if it is \emph{equivalent} to a serial execution under some permutation $\pi(T)$ of $T$. The permutation $\pi(T)$ is called the serializable order.
\end{definition}

\begin{definition}[Determinism]
	\label{def:Determinism}
	A concurrent control protocol is deterministic if multiple independent executions of $T$ on $S$ always produce the identical serializable order $\pi(T)$.
\end{definition}

In Definition~\ref{def:serializability}, ``equivalent'' means that, for every transaction, the concurrent execution and the corresponding serial execution produce identical read/write sets, thereby leading to the same final world state.

\subsection{Directed Graph}
\label{DirectedGraphBackground}

Lantern models transaction dependencies using a directed graph. 
Formally, a directed graph is represented as 
\begin{equation*}
	G=(V,E)
\end{equation*}
where $V$ and $E$ denote the sets of vertices and directed edges, respectively. 
For a vertex $v \in V$, its in-degree is the number of incoming edges pointing to $v$, whereas its out-degree is the number of outgoing edges originating from $v$. 
A vertex with an in-degree of zero is referred to as a source vertex, while a vertex with an out-degree of zero is called a sink vertex. 
A directed cycle exists if a directed path originates from a vertex and eventually returns to it. 
A directed graph containing no cycles is termed a directed acyclic graph (DAG). 
The problem of removing the minimum number of vertices to eliminate all cycles from a directed graph, thereby transforming it into a DAG, is known as the Feedback Vertex Set (FVS) problem \cite{FVS}, which is NP-hard.

\section{Lantern Design}
This section begins with an overview of the Lantern, followed by a detailed breakdown of its pipeline in Sections~\ref{TransactionSelectionandParallelExecution} through~\ref{TransactionCommitment}. 
Following that, Section~\ref{Back-PropagationMechanism} introduces a crucial performance optimization called the Back-Propagation mechanism, whose liveness is rigorously proven in Section~\ref{LivenessofBackPropagationPhase}. 
In Section~\ref{ConflictFreeBatchSelectionMechanism}, we introduce a Conflict‑Free Batch‑Selection mechanism for RMW‑intensive scenarios.
Finally, we provide a proof of Lantern's serializability in Section~\ref{SerializabilityandDeterminism}.

\subsection{Design Overview}
\label{DesignOverview}
In blockchain systems, a \emph{block} comprises a substantial number of transactions awaiting processing. 
As illustrated in Fig.~\ref{highlevelworkflow}, Lantern initially selects a fixed number of transactions with the smallest indices from the pending sequence to initiate a processing round.
Within each round, Lantern constructs a local dependency graph for these selected transactions and applies a Back-Propagation mechanism to identify a deterministic, committable set of transactions (shown in red).
The remaining, aborted transactions (shown in gray) are then sorted in ascending order of their indices and reinserted at the head of the pending sequence.
Lantern subsequently retrieves transactions from this updated sequence to start the next round.
This process loops continuously until the pending transaction sequence is exhausted, marking the completion of the block processing.
The following subsections describe each core component of Lantern in detail.
\begin{figure}[t]
	\centering
	\includegraphics[width=0.9\linewidth]{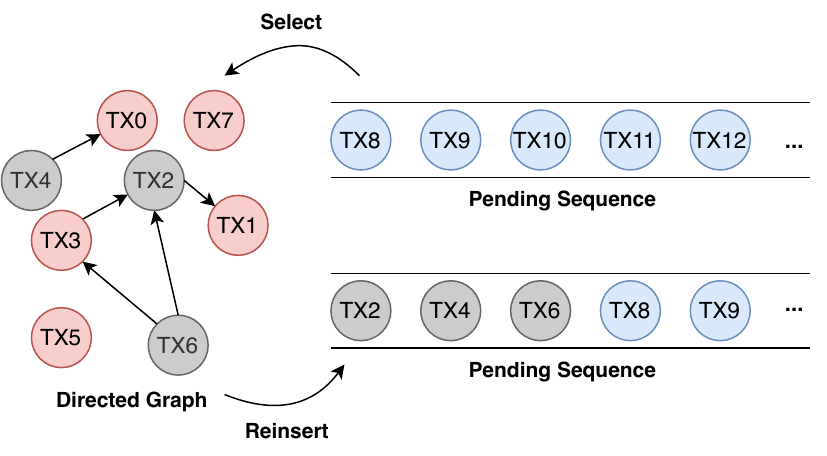}
	\caption{High-Level Transaction Processing of Lantern}
	\label{highlevelworkflow}
\end{figure}


\subsection{Transaction Selection and Parallel Execution}
\label{TransactionSelectionandParallelExecution}
Before executing transactions, all nodes in the network reach consensus on a globally ordered transaction sequence. Each node then independently executes this same sequence using Lantern.
To build a dependency graph that captures transaction dependencies, Lantern first executes transactions to extract their respective read and write sets, as follows:

\subsubsection{Selection Phase}
Let $B_{\mathit{size}}$ denote the number of transactions processed in each round. To maintain determinism, every node is configured with an identical batch size $B_{\mathit{size}}$ and adheres to a deterministic selection rule: the $B_{\mathit{size}}$ transactions with the lowest indices are popped from the pending transaction queue to form the current batch. Consequently, all nodes are guaranteed to select an identical transaction set.

\subsubsection{Execution Phase}
Following selection, each node executes all selected transactions in parallel against its current world state. During this speculative execution, write operations are strictly buffered rather than directly committed, ensuring that the underlying world state remains invariant. Upon completing the execution of a transaction $\mathit{TX}_i$, the node captures and caches its corresponding read set $\mathit{ReadSet(TX_i)}$ and write set $\mathit{WriteSet(TX_i)}$. 
Because all nodes start from an identical world state that remains unmutated throughout this phase, the read and write sets generated for each transaction are guaranteed to be identical across all nodes.

\subsection{DAG Construction}
\label{DependencyGraphProcessing}
Following the parallel execution of the selected transactions, data conflicts may arise that threaten execution serializability. Lantern focuses exclusively on \emph{Read-After-Write (RAW)} conflicts, which are formalized as follows:

\begin{definition}[RAW Conflict]
	\label{def:RAW}
	A transaction $\mathit{TX}_i$ has a RAW conflict with a preceding transaction $\mathit{TX}_j$ ($j < i$) if $\mathit{TX}_i$ reads a data item that is written by $\mathit{TX}_j$.
\end{definition}

Specifically, for each executed transaction, Lantern checks whether its read set conflicts with the write sets of preceding transactions (i.e., those with lower indices). Write operations by subsequent transactions (with higher indices) are not considered when constructing the dependency graph. As a result, the order in which committable transactions are serialized must be handled carefully, as discussed in Section~\ref{TransactionCommitment}.

To systematically capture RAW conflicts among transactions in the current batch, Lantern constructs a dependency graph according to Rule~\ref{rule:GraphConstruction}:

\begin{prule}[Graph Construction]
	\label{rule:GraphConstruction}
	Each executed transaction in the batch corresponds to a vertex in the dependency graph. A directed edge ($\mathit{TX}_i \rightarrow \mathit{TX}_j$) is added from $\mathit{TX}_i$ to $\mathit{TX}_j$ if and only if $\mathit{TX}_i$ has a RAW conflict with a preceding transaction $\mathit{TX}_j$ ($j < i$). 
	Self-loops are inherently excluded since a transaction cannot form an edge to itself ($j \neq i$).
\end{prule}

\begin{algorithm}[t]
	\caption{Dependency Graph Construction}
	\label{alg:dependency_graph_construction}
	
	\SetAlgoVlined
	
	\SetKwInOut{Input}{Input}
	\SetKwInOut{Output}{Output}
	
	\Input{executed transaction sequence $\mathcal{T}=(\mathit{TX}_0,\mathit{TX}_1,\ldots,\mathit{TX}_{n-1})$}
	\Output{dependency graph $\mathit{DAG}=(V,E)$}
	
	$V\leftarrow\{0,1,\ldots,n-1\}$\;
	$E\leftarrow\emptyset$\;
	
	\ForEach{$i\in\{0,1,\ldots,n-1\}$ \textbf{in parallel}}{
		\tcc{thread-local edge set for $\mathit{TX}_i$}
		$E_i \leftarrow \emptyset$\;
		\For{$j\leftarrow 0$ \KwTo $i-1$}{
			\ForEach{$w\in\mathit{WriteSet}(\mathit{TX}_j)$}{
				\If{$w\in\mathit{ReadSet}(\mathit{TX}_i)$}{
					$E_i\leftarrow E_i\cup\{(i,j)\}$\;
					\textbf{break}\tcp*{one conflict is sufficient}
				}
			}
		}
	}
	$E \leftarrow \bigcup_{i=0}^{n-1} E_i$\;
	\Return $\mathit{DAG}=(V,E)$\;
	
\end{algorithm}

Consequently, for every directed edge ($\mathit{TX}_i \rightarrow \mathit{TX}_j$) in the dependency graph, the start vertex always has a larger index than the end vertex ($i > j$). By construction, this strict topological ordering precludes any \emph{backward edges} (i.e., from lower-indexed to higher-indexed transactions), guaranteeing that the generated dependency graph is inherently \emph{acyclic}. In the remainder of this paper, we refer to this graph as a \emph{Directed Acyclic Graph (DAG)}.

To construct the DAG efficiently, Algorithm~\ref{alg:dependency_graph_construction} presents a parallel graph construction scheme utilizing the read sets $\mathit{ReadSet}(\mathit{TX}_i)$ and write sets $\mathit{WriteSet}(\mathit{TX}_i)$ cached during the Execution Phase. 
Because conflict detection for each transaction $\mathit{TX}_i$ depends solely on preceding transactions $\mathit{TX}_j$ ($j < i$), all transactions can evaluate their RAW dependencies fully in parallel (Line 3). 
To avoid lock contention during graph building, each worker thread maintains a thread-local outgoing edge set $E_i$ dedicated to $\mathit{TX}_i$ (Line 4). For each preceding transaction $\mathit{TX}_j$, thread $i$ scans $j$'s write set; upon detecting the first write-read overlap with $\mathit{ReadSet}(\mathit{TX}_i)$, the directed edge $(i, j)$ is inserted into $E_i$, and the inner loop immediately breaks to bypass redundant key checks for $\mathit{TX}_j$ (Lines 5–9). 
Because thread $i$ exclusively mutates its local outgoing edge set $E_i$, edge generation across threads is inherently collision-free and lockless. 
Finally, all thread-local edge sets are merged into $E = \bigcup_{i=0}^{n-1} E_i$ (Line 10) to yield the complete dependency graph $\mathit{DAG}=(V,E)$.

The graph construction process preserves determinism by ensuring that the constructed DAG is identical across all network nodes.
This inherently holds because transaction indices are globally identical, and speculative execution yields strictly identical read and write sets for each transaction.

\subsection{Transaction Commitment}
\label{TransactionCommitment}

\begin{figure}[t]
	\centering
	\includegraphics[width=0.9\linewidth]{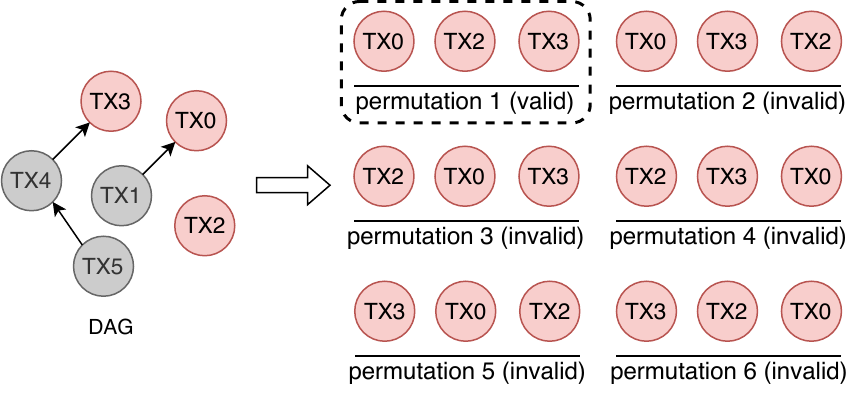}
	\caption{Only the ascending order of zero-out-degree transactions forms a serializable execution order for the current round.}
	\label{permutation}
\end{figure}

Following Section~\ref{DependencyGraphProcessing}, we obtain a DAG that precisely captures the RAW conflicts among transactions. Crucially, any transaction corresponding to a vertex with zero out-degree—including both sink vertices and isolated vertices—is guaranteed not to read any data item modified by its preceding transactions. Consequently, its read operations are completely independent of writes from its preceding transactions. Therefore, all vertices with zero out-degree can be safely committed in the current round.

The first challenge lies in determining which serial execution order is equivalent to the concurrent execution of these zero-out-degree transactions. As illustrated in Fig.~\ref{permutation}, three committable transactions (shown in red) yield six possible permutations. Because Lantern exclusively tracks each transaction's RAW conflicts against its preceding transactions, only the ascending order (Permutation 1 in Fig.~\ref{permutation}) guarantees a valid serializable execution order. 
This holds because under the ascending order, every transaction $\mathit{TX}_i$ is preceded solely by transactions $\mathit{TX}_j$ with smaller indices ($j < i$). Since $\mathit{TX}_i$ possesses a zero out-degree in the DAG, it is explicitly guaranteed to have no RAW dependencies on any preceding transaction $\mathit{TX}_j$. Consequently, executing them serially in ascending order yields the exact same read values as their parallel speculative execution against the initial world state.
Conversely, the DAG cannot guarantee the validity of any other permutation. 
Consider Permutation 2 $\langle \mathit{TX}_0, \mathit{TX}_3, \mathit{TX}_2 \rangle$, where $\mathit{TX}_3$ precedes $\mathit{TX}_2$. In a serial execution following this order, $\mathit{TX}_2$'s reads must observe any updates previously committed by $\mathit{TX}_3$. However, during the speculative Execution Phase, $\mathit{TX}_2$ executed against the initial, invariant world state—which omits $\mathit{TX}_3$'s writes. This may introduces a potential $\mathit{TX}_2 \rightarrow \mathit{TX}_3$ dependency. Since our DAG construction intentionally omits such backward edges to preclude cycles, we cannot verify whether such a potential conflict exists between $\mathit{TX}_3$ and $\mathit{TX}_2$, which would invalidate Permutation 2.
The same reasoning applies to the remaining four permutations, as each places at least one higher-indexed transaction before a lower-indexed one.

The second challenge lies in ensuring true equivalence to serial execution in ascending order. Regarding read operations, Lantern guarantees that no committable transaction reads any data item modified by its preceding transactions, thereby eliminating dirty reads. However, regarding write operations, write-after-write (WAW) conflicts may still occur when multiple transactions update the same data key. To maintain full equivalence to ascending serial execution, Lantern must preserve the overwrite semantics inherent to serial processing—specifically, when multiple transactions write to the same key, the transaction appearing later in the sequence must overwrite the values written by earlier ones. Therefore, the commit rule is:

\begin{prule}[Commit]
	\label{rule:Commit}
	For all committable transactions, sort them in ascending order of their indices and apply their write sets to the world state. If multiple transactions write to the same key, preserve only the value written by the transaction with the highest index.
\end{prule}

\begin{algorithm}[t]
	\caption{Transaction Commitment}
	\label{alg:transaction_commit}
	
	\SetAlgoVlined
	
	\SetKwInOut{Input}{Input}
	\SetKwInOut{Output}{Output}
	
	\Input{$DAG=(V,E)$; transaction set $\mathcal{T}$}
	\Output{Updated world state}
	
	$S \leftarrow \{v\in V \mid \mathrm{outDegree}(v)=0\}$\;
	
	sort $S$ in ascending transaction-index order\;
	
	\ForEach{$i\in S$}{
		\ForEach{$w\in \mathit{WriteSet(TX_i)}$}{
			apply $w$ to the world state\;
		}
	}
	
	\Return updated world state\;
	
\end{algorithm}

Algorithm~\ref{alg:transaction_commit} details our commitment procedure. By applying the write set of each transaction to the world state serially, this process naturally satisfies the required overwrite semantics. Since this state update phase is lightweight, we intentionally forgo parallelization techniques—such as Compare-and-Swap (CAS)—to eliminate unnecessary concurrency overhead.

This commitment phase strictly guarantees network-wide determinism across iterative rounds. This holds because all nodes construct an identical DAG and deterministically derive the exact same set of zero-out-degree transactions for commitment. With globally uniform write sets, the resulting world state remains perfectly consistent across all nodes, serving as the identical initial state for the next round.
Concurrently, uncommitted transactions (represented by gray vertices in Fig.~\ref{permutation}) are sorted in ascending order of their indices and reinserted at the head of the pending queue, yielding an updated sequence that is likewise identical across all nodes.
By ensuring that every round begins with an identical world state and pending transaction sequence, Lantern guarantees deterministic transaction processing across all nodes.

\subsection{Back-Propagation Mechanism}
\label{Back-PropagationMechanism}
Sections~\ref{TransactionSelectionandParallelExecution} through \ref{TransactionCommitment} detail the baseline execution pipeline of Lantern. However, restricting commitments strictly to zero-out-degree vertices is inherently conservative, prematurely aborting conflict-free transactions that reside deeper in the dependency chain. To overcome this limitation, we introduce a novel \emph{Back-Propagation (BP)} mechanism that systematically salvages additional committable transactions while strictly preserving network-wide determinism.

To establish the foundation of status propagation, we first clarify the precise semantics of dependency edges.
A directed edge $\mathit{TX}_i \rightarrow \mathit{TX}_j$ indicates a RAW dependency of $\mathit{TX}_i$ on $\mathit{TX}_j$.
Since all transactions are executed against the same initial world state during the Execution Phase, committing $\mathit{TX}_j$ in the current round invalidates the execution of $\mathit{TX}_i$ due to a ``dirty read'', forcing $\mathit{TX}_i$ to abort. Conversely, for $\mathit{TX}_i$ to safely commit, $\mathit{TX}_j$ must be aborted to ensure that the state read by $\mathit{TX}_i$ remains ``clean''.

Consider the dependency graph in Fig.~\ref{bpexample}. Under the baseline approach, only the two sink vertices ($\mathit{TX}_0$ and $\mathit{TX}_2$) are committed, while all remaining transactions are aborted. Specifically, because $\mathit{TX}_1$ and $\mathit{TX}_3$ directly depend on these two sink transactions, they must be aborted to avoid dirty reads. Crucially, once $\mathit{TX}_1$ and $\mathit{TX}_3$ are determined to be aborted, their write operations are discarded, eliminating the dependency edges they induce on upstream transactions.

\begin{figure}[t]
	\centering
	\includegraphics[width=\linewidth]{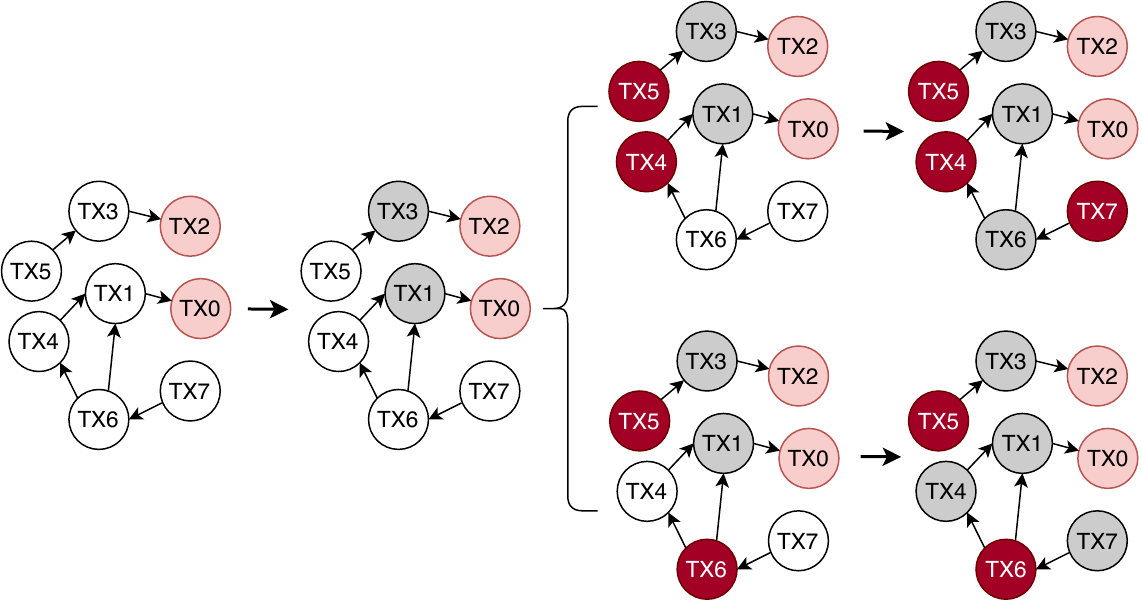}
	\caption{A motivating example demonstrating the potential of transaction salvage and the issue of nondeterminism.}
	\label{bpexample}
\end{figure}

As illustrated in the upper branch of Fig.~\ref{bpexample}, if we salvage transactions along the upper path, $\mathit{TX}_4$ and $\mathit{TX}_5$ become eligible for commitment alongside the baseline sinks. Furthermore, committing $\mathit{TX}_4$ forces its reader $\mathit{TX}_6$ to abort to prevent a dirty read, which subsequently frees $\mathit{TX}_7$ for commitment. In total, five transactions—$\{\mathit{TX}_0, \mathit{TX}_2, \mathit{TX}_4, \mathit{TX}_5, \mathit{TX}_7\}$—can be safely committed. 
Because these salvaged transactions are guaranteed to be RAW-conflict-free with respect to their committable predecessors, they can be safely committed according to Rule~\ref{rule:Commit} without violating serializability.

However, as shown in the lower branch of Fig.~\ref{bpexample}, an alternative path emerges once $\mathit{TX}_1$ is determined to be aborted.
Specifically, the dependency edge restricting $\mathit{TX}_6$ also disappears. If we instead mark $\mathit{TX}_6$ as committable, $\mathit{TX}_4$ must be aborted to ensure that $\mathit{TX}_6$'s read set remains uncontaminated. Concurrently, $\mathit{TX}_7$ must also be aborted to prevent a dirty read from $\mathit{TX}_6$. This path yields a completely different set of committable transactions.

While both branches successfully salvage more transactions than the baseline, the existence of multiple valid execution paths introduces severe nondeterminism. 
Thus, the core challenge lies in designing a mechanism that expands the committable set while preserving determinism across all distributed nodes.
To achieve this, Lantern adopts a novel Back-Propagation mechanism inspired by error backpropagation in neural networks. Starting from the sink vertices, status markings propagate backward toward the source vertices according to four deterministic rules:

\begin{prule}[Back-Propagation]
	\label{rule:BackPropagation}
	\begin{enumerate}[label=(\roman*)]
		\item \textbf{Initialization (Red):} All zero-out-degree vertices in the DAG are initially marked red.
		\item \textbf{Conflict Propagation (Gray):} Any unmarked vertex with a directed edge pointing to a red vertex is marked gray.
		\item \textbf{Salvage Propagation (Red):} Any unmarked vertex whose successors are \emph{exclusively} marked gray is marked red.
		\item \textbf{Termination:} Steps (ii) and (iii) are repeated iteratively in alternating phases until no further state transitions occur.
	\end{enumerate}
\end{prule}

Specifically, (i)~\emph{Initialization} directly corresponds to the baseline commit scheme presented in Section~\ref{TransactionCommitment}. 
(ii)~\emph{Conflict Propagation} identifies transactions that must abort because they would incur dirty reads from guaranteed committable transactions. 
(iii)~In \emph{Salvage Propagation}, the condition ``exclusively'' is paramount: if an unmarked vertex points to both a gray vertex and an unresolved unmarked vertex, its status remains ambiguous because its unmarked successor's outcome is still unresolved. Requiring all successors to be gray guarantees that no committed predecessor will invalidate this vertex's read set, ensuring it can safely commit.

\begin{algorithm}[t]
	\caption{Back-Propagation}
	\label{alg:back-propagation}
	
	\SetAlgoVlined
	
	\SetKwInOut{Input}{Input}
	\SetKwInOut{Output}{Output}
	
	\Input{Directed acyclic graph $G = (V,E)$}
	\Output{Committable transaction set $\mathcal{R}$; Uncommittable transaction set $\mathcal{G}$}
	
	\ForEach{$v \in V$}{
		$\mathit{mark}[v] \leftarrow \textit{unmarked}$\;
	}
	
	$\mathit{layer} \leftarrow \emptyset$\;
	
	\ForEach{$v \in V$}{
		\If{$outDegree(v)=0$}{
			$\mathit{mark}[v] \leftarrow red$\;
			add $v$ to $\mathit{layer}$\;
		}
	}
	
	$\mathit{isRedLayer} \leftarrow true$\;
	
	\While{there exists unmarked vertices in $V$}{
		
		$\mathit{candidates} \leftarrow \emptyset$\;
		
		\ForEach{$v \in \mathit{layer}$}{
			\ForEach{$u \in reverseNeighbors(v)$}{
				\If{$\mathit{mark}[u] = \textit{unmarked}$}{
					add $u$ to $\mathit{candidates}$\;
				}
			}
		}
		
		$\mathit{nextLayer} \leftarrow \emptyset$\;
		
		\eIf{$\mathit{isRedLayer} = true$}{
			
			\ForEach{$u \in \mathit{candidates}$}{
				$\mathit{mark}[u] \leftarrow gray$\;
				add $u$ to $\mathit{nextLayer}$\;
			}
			
		}{
			
			$\mathit{redCandidates} \leftarrow \emptyset$\;
			
			\ForEach{$u \in \mathit{candidates}$}{
				
				$\mathit{allMarked} \leftarrow true$\;
				
				\ForEach{$w \in neighbors(u)$}{
					\If{$\mathit{mark}[w] = \textit{unmarked}$}{
						$\mathit{allMarked} \leftarrow false$\;
						\textbf{break}\;
					}
				}
				
				\If{$\mathit{allMarked} = true$}{
					add $u$ to $\mathit{redCandidates}$\;
				}
			}
			
			\ForEach{$u \in \mathit{redCandidates}$}{
				$\mathit{mark}[u] \leftarrow red$\;
				add $u$ to $\mathit{nextLayer}$\;
			}
		}
		
		$\mathit{layer} \leftarrow \mathit{nextLayer}$\;
		$\mathit{isRedLayer} \leftarrow \neg \mathit{isRedLayer}$\;
	}
	
	$\mathcal{R} \leftarrow \{v \in V \mid \mathit{mark}[v]=red\}$\;
	$\mathcal{G} \leftarrow \{v \in V \mid \mathit{mark}[v]=gray\}$\;
	
	\Return $\mathcal{R}, \mathcal{G}$\;
\end{algorithm}

\begin{figure*}[t]
	\centering
	
	\begin{subfigure}{0.175\textwidth}
		\includegraphics[width=\linewidth]{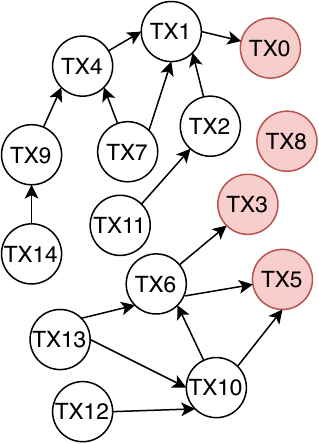}
		\caption{Round 1}
		\label{Round1}
	\end{subfigure}
	\hfill 
	\begin{subfigure}{0.175\textwidth}
		\includegraphics[width=\linewidth]{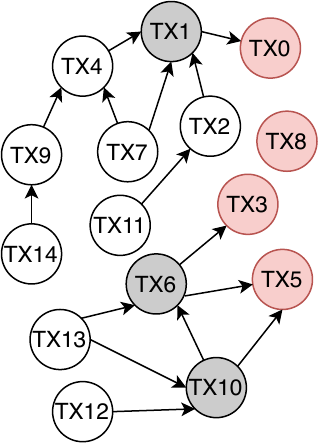}
		\caption{Round 2}
		\label{Round2}
	\end{subfigure}
	\hfill
	\begin{subfigure}{0.175\textwidth}
		\includegraphics[width=\linewidth]{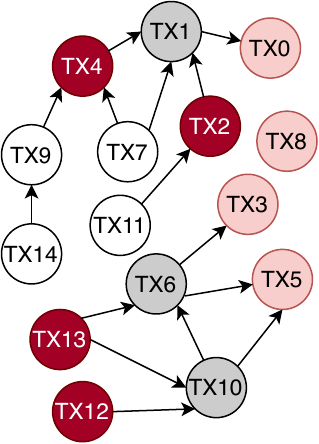}
		\caption{Round 3}
		\label{Round3}
	\end{subfigure}
	\hfill
	\begin{subfigure}{0.175\textwidth}
		\includegraphics[width=\linewidth]{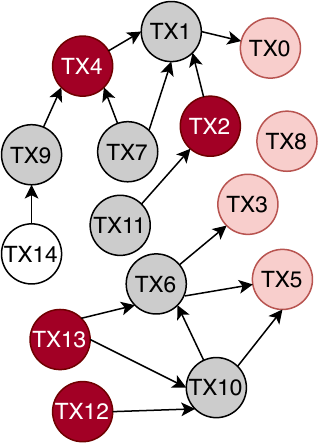}
		\caption{Round 4}
		\label{Round4}
	\end{subfigure}
	\hfill
	\begin{subfigure}{0.175\textwidth}
		\includegraphics[width=\linewidth]{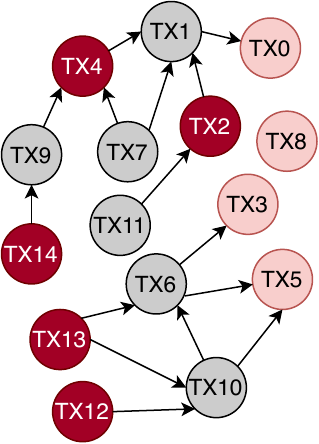}
		\caption{Round 5}
		\label{Round5}
	\end{subfigure}
	
	\caption{An Example of Back-Propagation in Lantern}
	\label{fig:full_width_results}
\end{figure*}

Algorithm~\ref{alg:back-propagation} formalizes this iterative layer-by-layer procedure. In each iteration, Lantern first collects an unmarked candidate set $\mathit{candidates}$ comprising the reverse neighbors (i.e., predecessors) of the current layer (Lines 9--14). When the current layer is red, all vertices in $\mathit{candidates}$ are directly marked gray (Lines 16--19). 

Conversely, when the current layer is gray, if all successors of a candidate vertex are confirmed to be marked, those successors are guaranteed to be exclusively gray, allowing the candidate vertex to be safely marked red (Lines 22--29).
This guarantee holds because candidates are derived solely from the current gray layer; thus, a candidate vertex can point only to gray or unmarked vertices, but never to a red one. Had it possessed a directed edge to a previously generated red vertex, it would have been unconditionally marked gray in the iteration immediately following that red layer's processing, directly contradicting its current unmarked status.
When candidate vertices qualify to be marked red, they are first collected into a temporary set, $\mathit{redCandidates}$ (Lines 28--29), rather than being marked immediately. The actual state transitions are performed only after all candidate vertices in the current layer have been fully evaluated (Lines 30--32). This deferred-marking strategy ensures that every decision within the current propagation layer is evaluated against a static and unchanged graph state. Without this deferral, immediate state updates would allow newly marked red vertices to improperly influence the evaluation of remaining candidates in the current layer. Such premature updates would not only compromise algorithm correctness by causing erroneous transitions, but also introduce nondeterminism when candidate vertices are processed in an arbitrary traversal order.

Fig.~\ref{fig:full_width_results} illustrates a step-by-step example of Back-Propagation. 
In Round 1, four zero-out-degree vertices are initially marked red. 
In Round 2, Lantern identifies the reverse neighbors of this initial layer, forming the candidate set $\{\mathit{TX}_1, \mathit{TX}_6, \mathit{TX}_{10}\}$. 
Because the preceding layer was red, all vertices in this candidate set are directly marked gray without further evaluation.
In Round 3, the candidates $\{\mathit{TX}_2, \mathit{TX}_4, \mathit{TX}_7, \mathit{TX}_{12}, \mathit{TX}_{13}\}$ are evaluated. Since $\mathit{TX}_7$ points to $\mathit{TX}_4$, which remains unmarked, $\mathit{TX}_7$ retains its unmarked status. The remaining candidates are marked red because their successors are exclusively gray. 
In Round 4, $\mathit{TX}_7$ is again collected into the candidate set $\{\mathit{TX}_7, \mathit{TX}_9, \mathit{TX}_{11}\}$, and all candidates are marked gray.
Finally, in Round 5, $\mathit{TX}_{14}$ is marked red because its only successor is gray.

Back-Propagation guarantees network-wide determinism, ensuring that every node produces an identically colored DAG.
This holds because the DAG topology is identical across all participating nodes. Specifically, all nodes initialize the same set of red vertices. Inductively, because each subsequent marking round depends solely on the deterministic state of preceding layers, both (1) gray vertices assigned under Conflict Propagation and (2) red vertices assigned under Salvage Propagation remain strictly identical across nodes. By induction, this state consistency propagates throughout all iterations, guaranteeing an identical coloring outcome for the DAG.

Once Back-Propagation completes, all transactions marked red are collected as the final committable set. These transactions are sorted in ascending order of their indices and serially applied to the world state following Rule~\ref{rule:Commit}, while uncommitted gray transactions are sorted in ascending order of their indices and reinserted at the head of the pending queue.

The proposed Back-Propagation mechanism increases the number of committable transactions while strictly preserving determinism across all participating nodes. 
However, a critical question arises: could this iterative marking process encounter a liveness bottleneck, causing the algorithm to stall or fail to terminate? 
To establish the theoretical soundness of our design and address this concern, we formally prove the \emph{liveness} of the Back-Propagation mechanism in Section~\ref{LivenessofBackPropagationPhase}.

\subsection{Liveness of Back-Propagation}
\label{LivenessofBackPropagationPhase}
We first introduce two supporting lemmas to dissect the iterative marking mechanism before proving the main theorem.

\begin{lemma}
	\label{lem:nonempty_candidates}
	In any iteration of the Back-Propagation process, as long as the set of unmarked vertices in the DAG is non-empty, the generated $\mathit{candidates}$ set is guaranteed to be non-empty.
\end{lemma}

\begin{proof}
	During initialization, all global sink vertices and isolated vertices  in DAG are marked red, meaning any remaining unmarked vertices must have outgoing edges.
	Let $\mathbf{U}$ be the non-empty set of unmarked vertices, and let $\mathbf{L_k}$ be the current $\mathit{layer}$ serving as the active frontier. Since DAG is a acyclic, the subgraph $G'[\mathbf{U}]$ induced by $\mathbf{U}$ must also be a DAG, which guarantees the existence of at least one local sink vertex $u \in \mathbf{U}$. By definition, $u$ has no outgoing edges targeting any other vertex within $G'[\mathbf{U}]$. 
	However, since $u$ is a unmarked vertex in DAG which must have outgoing edges, all successors of $u$ must reside outside of $\mathbf{U}$, meaning they have already been marked in previous iterations.
	
	We now prove by contradiction that $u$ must possess at least one directed edge pointing to the current layer $\mathbf{L_k}$. Suppose $u$ does not point to $\mathbf{L_k}$. Then all successors of $u$ must belong to strictly older historical layers. Let $\mathbf{L_t}$ ($t < k$) be the most recently generated layer among all successors of $u$. During the historical iteration when $\mathbf{L_t}$ served as the active layer, $u$ would have been fetched into the $\mathit{candidates}$ set. At that time, we evaluate the state transition of $u$:
	\begin{enumerate}[label=(\roman*)]
		\item {If $\mathbf{L_t}$ was a red layer:} $u$ would be directly marked gray.
		\item {If $\mathbf{L_t}$ was a gray layer:} 
		Because $\mathbf{L_t}$ is the latest processed layer among the successors of $u$, all successors of $u$ have already been marked. Moreover, none of these successors can be red. Otherwise, according to the Conflict Propagation rule (Rule~\ref{rule:BackPropagation}(ii)), $u$ would have been marked gray immediately after the layer containing that red successor was processed, contradicting the fact that $u$ remains unmarked. Therefore, all successors of $u$ must be gray, and $u$ can be safely marked red.
	\end{enumerate}
	In either case, $u$ would have been successfully marked during that historical iteration. This directly contradicts the established premise that $u$ is currently an unmarked vertex. Therefore, the assumption is false. $u$ must possess at least one successor in the current layer $\mathbf{L_k}$. 
	As a result, $u$ will inevitably be fetched into $\mathit{candidates}$, ensuring $\mathit{candidates} \neq \emptyset$.
\end{proof}

\begin{lemma}
	\label{lem:mark_at_least_one}
	Given a non-empty $\mathit{candidates}$ set in any iteration, at least one candidate within this set is guaranteed to be marked.
\end{lemma}

\begin{proof}
	We evaluate the two propagation branches based on the current $\mathit{layer}$ $\mathbf{L_t}$:
	\begin{enumerate}[label=(\roman*)]
		\item {If $\mathbf{L_t}$ was a red layer:} All vertices in $\mathit{candidates}$ are unconditionally marked gray. Since $\mathit{candidates} \neq \emptyset$, at least one vertex transitions to the gray state.
		
		\item {If $\mathbf{L_t}$ was a gray layer:} First, any vertex in $\mathit{candidates}$ can only point to gray or unmarked vertices. Second, recall the local sink vertex $u$ of the global unmarked subgraph $G'[\mathbf{U}]$ from Lemma~\ref{lem:nonempty_candidates}, which is guaranteed to be fetched into $\mathit{candidates}$. As the sink of $G'[\mathbf{U}]$, $u$ has no outgoing edges to any unmarked vertices. Combining these two constraints, the outgoing edges of $u \in \mathit{candidates}$ must exclusively point to gray vertices. Therefore, $u$ is guaranteed to be marked red.
	\end{enumerate}
	In both cases, at least one unmarked candidate vertex successfully transitions to a marked state.
\end{proof}

\begin{theorem}[Liveness]
	\label{thm:termination}
	If the dependency graph $G=(V,E)$ is acyclic (DAG), the Back-Propagation process satisfies the liveness property and is guaranteed to terminate within a finite number of iterations.
\end{theorem}

\begin{proof}
	During initialization, all global zero-out-degree vertices of DAG are successfully marked red, establishing the initial non-empty $\mathit{layer}$. In each subsequent iteration, if unmarked vertices remain, Lemma~\ref{lem:nonempty_candidates} guarantees that the fetched $\mathit{candidates}$ set is non-empty. Subsequently, Lemma~\ref{lem:mark_at_least_one} guarantees that at least one vertex from this $\mathit{candidates}$ set transitions from an unmarked to a marked state. 
	Given that the total number of vertices $|V|$ is finite and each vertex changes its state at most once, therefore, the algorithm must terminate after at most $|V|$ iterations.
	The upper bound of $\vert{}V\vert{}$ iterations corresponds to the worst-case scenario where the DAG forms a linear chain ($\mathit{TX}_0 \leftarrow \mathit{TX}_1 \leftarrow \dots \leftarrow \mathit{TX}_{\vert{}V\vert{}-1}$).
\end{proof}

\subsection{Conflict‑Free Batch Selection Mechanism}
\label{ConflictFreeBatchSelectionMechanism}

While the Back-Propagation mechanism effectively finds more committable transactions when the DAG is deep, it remains constrained in \emph{Read-Modify-Write (RMW)} dominant scenarios. 
In an RMW operation, a transaction first reads a data item's current state, performs internal business logic, and subsequently writes the updated state back to the same item. 
A common real-world example is found in financial applications:
\begin{itemize}
	\item Deposit/Withdrawal: A transaction reads account $A$'s balance $\mathit{bal}_A$, adds or subtracts an amount $v$, and writes the new balance $(\mathit{bal}_A \pm v)$ back to $A$.
	\item Transfer: A transaction reads the balances of both sender $A$ ($\mathit{bal}_A$) and receiver $B$ ($\mathit{bal}_B$), and writes the updated values $(\mathit{bal}_A - v)$ and $(\mathit{bal}_B + v)$ back to $A$ and $B$, respectively.
\end{itemize}

In Lantern's standard batch processing pattern, if multiple transactions within a single batch target the same hot account (e.g., account $A$) with RMW operations, only a single transaction can be committed. 
As illustrated in Fig.~\ref{rmw}, because each transaction performs both read and write operations on account $A$'s balance, every transaction in the batch develops a RAW dependency on \emph{all} of its preceding transactions. 
This results in a shallow, star-like DAG where only the first transaction $\mathit{TX}_0$ can be safely committed; according to the Conflict Propagation rule (Rule~\ref{rule:BackPropagation}(ii)), all subsequent transactions must be aborted to prevent dirty reads from $\mathit{TX}_0$.

\begin{figure}[t]
	\centering
	\includegraphics[width=0.45\linewidth]{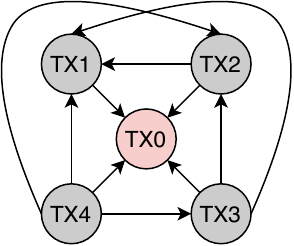}
	\caption{An example of a shallow star-like DAG under hot-account RMW operations, where only $\mathit{TX}_0$ is committable.}
	\label{rmw}
\end{figure}

To resolve this issue, we propose the \emph{Conflict-Free Batch Selection (CFBS)} mechanism specifically for RMW-intensive workloads. 
The core insight of CFBS is to proactively avoid conflicts during the batch selection phase rather than reactively aborting transactions after execution. 
Before execution, a transaction's accessed accounts can be statically extracted from its payload, as the input parameters for smart contract invocation explicitly specify the participating accounts.
By leveraging this extracted account information, CFBS ensures that no two transactions selected within the same batch operate on overlapping accounts.

Algorithm~\ref{alg:cfbs} details the CFBS procedure. Rather than blindly popping a fixed number of $B_{\mathit{size}}$ transactions with the smallest indices, CFBS sequentially scans the pending transaction sequence $\mathcal{P}$ in ascending index order. For each transaction $\mathit{TX_i}$, CFBS extracts its target accounts $\mathit{Accounts}_{\mathit{TX_i}}$ (Line~5). If these accounts do not overlap with the occupied account set $\mathcal{O}$ (Line~6), $\mathit{TX_i}$ is added to the selected set $\mathcal{S}$, and its target accounts are merged into $\mathcal{O}$ (Lines~7--8). Otherwise, $\mathit{TX_i}$ is deferred and appended to the remaining transaction sequence $\mathcal{P}'$ (Line~10). Consequently, the effective batch size becomes dynamic, defined as $B_{\mathit{size}} = |\mathcal{S}|$.

CFBS strictly preserves determinism across the network. Because the targeted accounts extracted from each transaction's payload are deterministic, and all nodes execute Algorithm~\ref{alg:cfbs} over the identical pending sequence $\mathcal{P}$ in the exact same ascending order, every node independently derives the identical selected batch $\mathcal{S}$ for each round.

\begin{algorithm}[t]
	\caption{Conflict-Free Batch Selection}
	\label{alg:cfbs}
	
	\SetAlgoVlined
	
	\SetKwInOut{Input}{Input}
	\SetKwInOut{Output}{Output}
	
	\Input{pending transaction sequence $\mathcal{P}$}
	\Output{selected transactions $\mathcal{S}$; remaining transaction sequence $\mathcal{P}'$}
	
	$\mathcal{S}\leftarrow\emptyset$\;
	$\mathcal{P}'\leftarrow\emptyset$\;
	$\mathcal{O}\leftarrow\emptyset$\tcp*{occupied accounts}
	
	\ForEach{$\mathit{TX}_i\in\mathcal{P}$ in ascending order of $i$}{
		$\mathit{Accounts}_{\mathit{TX}_i}
		\leftarrow\mathrm{ExtractAccounts}(\mathit{TX}_i)$\;
		
		\If{$\mathit{Accounts}_{\mathit{TX}_i}
			\cap\mathcal{O}=\emptyset$}{
			$\mathcal{S}\leftarrow
			\mathcal{S}\cup\{\mathit{TX}_i\}$\;
			$\mathcal{O}\leftarrow
			\mathcal{O}\cup
			\mathit{Accounts}_{\mathit{TX}_i}$\;
		}
		\Else{
			$\mathcal{P}'\leftarrow
			\mathcal{P}'\mathbin{\|}\mathit{TX}_i$\;
		}
	}
	
	\Return $\mathcal{S},\mathcal{P}'$\;
	
\end{algorithm}

\subsection{Serializability of Lantern}
\label{SerializabilityandDeterminism}

Theorem~\ref{thm:intra-batch} establishes the equivalent serializable order within a single batch, while Theorem~\ref{thm:inter-batch} extends this order across consecutive batches.

\begin{theorem} [Intra-Batch Serializability]
	\label{thm:intra-batch}
	Within a single batch, the equivalent serializable order of all committed transactions in Lantern is determined by their transaction indices in ascending order.
\end{theorem}

\begin{proof}
	On the read side, following the Back-Propagation mechanism, each committable transaction has no RAW dependencies on its preceding committable transactions. Therefore, when executed serially in ascending order of their indices, no committable transaction reads values written by its preceding committable transactions. Consequently, all read operations in this serial execution access the initial world state, yielding the same read values as those obtained during parallel speculative execution in the Execution Phase.
	On the write side, when multiple transactions write to the same key, Lantern preserves only the write from the transaction with the highest index. This behavior is equivalent to the overwrite semantics produced by executing all committable transactions serially in ascending order of their indices.
\end{proof}

\begin{theorem}[Inter-Batch Serializability]
	\label{thm:inter-batch}
	The equivalent serializable order across multiple batches in Lantern is obtained by sequentially composing the serializable orders of individual batches in ascending order of their batch indices.
\end{theorem}

\begin{proof}
	Under the Lantern framework, execution across batches enforces strict isolation. Specifically, a subsequent batch $B_{i+1}$ starts only after all committable transactions in $B_i$ complete and their write sets are applied to the world state.
	Viewing each batch as a ``macro-transaction'', the state transitions between batches are logically equivalent to a serial execution of these macro-transactions.
	Thus, concatenating the internal serializable orders of all batches in ascending order of their batch indices yields the equivalent serializable order of the overall execution.
\end{proof}

\begin{example}
	\label{ex:constructive_order}
	Consider an execution consisting of multiple batches, with a batch size of $B_{\mathit{size}} = 5$:
	\begin{itemize}
		\item Batch 1: Lantern processes the initial batch $\mathcal{B}_1 = \{\mathit{TX}_0, \mathit{TX}_1, \mathit{TX}_2, \mathit{TX}_3, \mathit{TX}_4\}$ and commits the subset $\mathcal{S}_1 = \langle \mathit{TX}_0, \mathit{TX}_2, \mathit{TX}_3 \rangle$. The aborted transactions $\{\mathit{TX}_1, \mathit{TX}_4\}$ are reinserted at the head of the pending queue.
		\item Batch 2: Lantern fetches the next batch $\mathcal{B}_2 = \{\mathit{TX}_1, \mathit{TX}_4, \mathit{TX}_5, \mathit{TX}_6, \mathit{TX}_7\}$ and commits the subset $\mathcal{S}_2 = \langle \mathit{TX}_1, \mathit{TX}_4, \mathit{TX}_6 \rangle$.
	\end{itemize}
	By sequentially concatenating the commit sequences across all batches, we construct the equivalent serializable order $\pi(T)$ of the overall execution as follows:
	\begin{equation*}
		\pi(T) = \mathcal{S}_1 \mathbin{\Vert} \mathcal{S}_2 \mathbin{\Vert} \dots = \langle \underbrace{\mathit{TX}_0, \mathit{TX}_2, \mathit{TX}_3}_{\text{Batch 1 ($\mathcal{S}_1$)}}, \underbrace{\mathit{TX}_1, \mathit{TX}_4, \mathit{TX}_6}_{\text{Batch 2 ($\mathcal{S}_2$)}}, \dots \rangle
	\end{equation*}
\end{example}

Crucially, Lantern's serializability proof is constructive, enabling the explicit determination of the equivalent serializable order $\pi(T)$ after all batches have been processed.
This explicit construction enables direct experimental verification of both serializability and determinism, as detailed in Section~\ref{CorrectnessValidation}. Beyond experimental validation, it also provides auditability and traceability for production blockchains, allowing external parties to independently verify state transitions by serially replaying transactions in the $\pi(T)$ order. In contrast, Aria's serializability proof is purely existential: it establishes that an equivalent serializable schedule exists but does not explicitly construct the corresponding serializable order. This non-constructive nature prevents direct experimental verification of its serializability through explicit serial replay.

\section{Evaluation}

\subsection{Experimental Setup}
\textbf{Hardware Platform.}
All experiments are performed on a dedicated Linux server equipped with dual Intel(R) Xeon(R) Gold 5218 processors, offering a total of 32 physical CPU cores across two NUMA nodes (16 cores and 128 GB DDR4 local memory per node). 
The server runs CentOS Linux 7 (kernel version 3.10.0). 
Simultaneous Multithreading (SMT) is disabled in the BIOS to eliminate hardware thread contention.

\textbf{Software Platform.}
All evaluations are integrated into ChainMaker v2.3.8.
Smart contracts are authored in Rust and executed on the Wasmer virtual machine.
To ensure fairness, Lantern and all baseline protocols are compiled and executed using Go 1.24 with default compiler optimization settings.
The system adopts TBFT consensus \cite{tendermint}, with the system configured to accommodate up to 1,000 transactions per block and a transaction pool of up to 50,000 transactions.
To fully saturate the execution engine, the ChainMaker Go SDK client is deployed on a separate dedicated server, spawning concurrent goroutines to submit workload transactions via gRPC.

\textbf{Baselines.}
We evaluate Lantern against two representative categories of concurrency control baselines:
\begin{itemize}[leftmargin=*]
	\item \emph{ChainMaker:} ChainMaker's default nondeterministic protocol based on Optimistic Concurrency Control (OCC) \cite{OCC}. It operates under the two-stage execution architecture (Fig.~\ref{towstageexecution}), where total execution-layer latency is the combined duration of the primary execution stage and the replica replay stage.
	\item \emph{Aria:} A state-of-the-art deterministic protocol operating without a priori knowledge \cite{aria}. Since original work has demonstrated Aria's superiority over classical deterministic protocols (such as Calvin \cite{Calvin}, Bohm \cite{BOHM}, PWV \cite{PWV}, and PB), we select it as our primary baseline and implement all optimization techniques proposed in its original paper.
\end{itemize}
Unless otherwise specified, the batch size $B_{\mathit{size}}$ is set to $5\times$ the number of physical CPU cores for both Lantern and Aria.

\begin{figure}[t]
	\centering
	\includegraphics[width=0.7\linewidth]{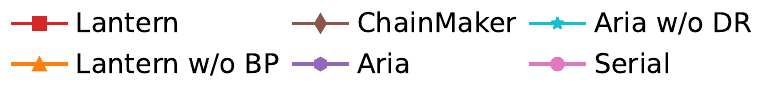}
	
	\vspace{1mm}
	
	\includegraphics[width=\linewidth]{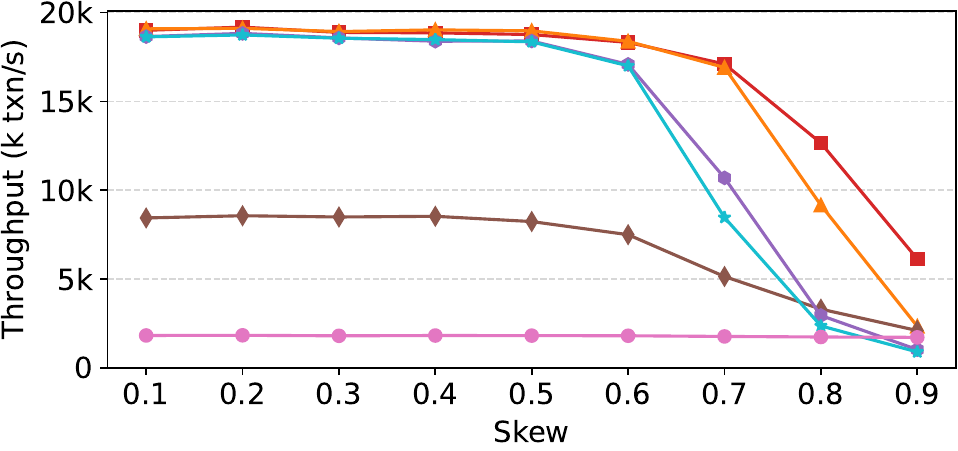}
	\caption{Overall Throughput on YCSB}
	\label{fig:OverallThroughputonYCSB}
\end{figure}

\textbf{Workloads.}
We use two standard benchmarks:
\begin{itemize}[leftmargin=*]
	\item \emph{YCSB \cite{YCSB}:}
	A micro-benchmark for fine-grained state access.
	Each transaction invokes a smart contract executing 10 operations (5 reads and 5 writes) over a key space of one million records. Read and write keys are independently sampled from a Zipfian distribution. Within a single transaction, all read keys (and similarly, write keys) are unique. Intersections between a transaction's read set and write set are allowed.
	\item \emph{SmallBank \cite{SMALLBANK}:} 
	A macro-benchmark characterized by read-modify-write (RMW) transactions.
	Our implementation includes three tables (\texttt{Account}, \texttt{Saving}, and \texttt{Checking}) initialized with 100,000 accounts. Each account starts with 100,000 tokens in both its saving and checking balances. We evaluate the five state-modifying transaction types: \texttt{send\_payment} (40\%), \texttt{deposit\_checking} (15\%), \texttt{transact\_saving} (15\%), \texttt{amalgamate} (15\%), and \texttt{write\_check} (15\%). Off-chain read-only queries (\texttt{balance}) and state initialization transactions (\texttt{create\_account}) are excluded from throughput measurements.
\end{itemize}

\textbf{Measurement Methodology.}
All evaluations strictly focus on execution-layer throughput.
To measure steady-state performance, boundary blocks are filtered out: initial warm-up/initialization blocks (the first 5 blocks in YCSB and account creation blocks in SmallBank) as well as the trailing 50 draining blocks are omitted.
Consequently, all reported metrics are calculated exclusively from the intermediate stable execution phase, with each metric averaged over \emph{at least 100 valid blocks}.

\subsection{Correctness Validation}
\label{CorrectnessValidation}
Using the serializable order $\pi(T)$ derived in Section~\ref{SerializabilityandDeterminism}, we validate Lantern's correctness over 10,000 continuous blocks under high contention.
Our validation comprises two phases:
First, to verify determinism, each block is executed via Lantern in multiple independent runs; we observe that all runs consistently yield an identical serializable order.
Second, to verify serializability, we sequentially re-execute the transactions according to the generated order $\pi(T)$. 
The evaluation results show that the read/write sets of each transaction perfectly match those from the parallel execution, and the final state remains identical between the serial and parallel executions. 
These empirical results demonstrate both the determinism and serializability of Lantern. 

\subsection{Overall Performance}
On YCSB, we disable the CFBS mechanism because the workload does not feature RMW operations. We evaluate Lantern alongside its variant Lantern w/o BP (without Back-Propagation) and four baseline systems: ChainMaker, Aria, Aria w/o DR (without deterministic reordering), and Serial execution.
On SmallBank, an RMW-intensive benchmark, the complete Lantern enables both BP and CFBS mechanisms by default.
We conduct a comprehensive ablation study using four variants: full Lantern, Lantern w/o BP (with CFBS only), Lantern w/o CFBS (with BP only), and Lantern w/o \{BP, CFBS\} (disabling both mechanisms).
For both benchmarks, we systematically vary the Zipfian \textit{skew} parameter to simulate different levels of data contention.

\begin{figure}[t]
	\centering
	\includegraphics[width=1.1\linewidth]{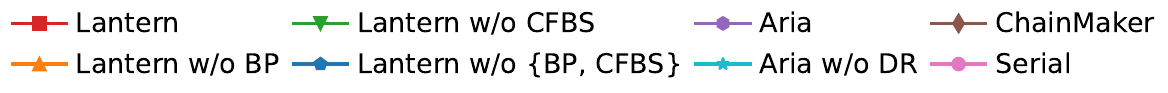}
	
	\vspace{1mm}
	
	\includegraphics[width=\linewidth]{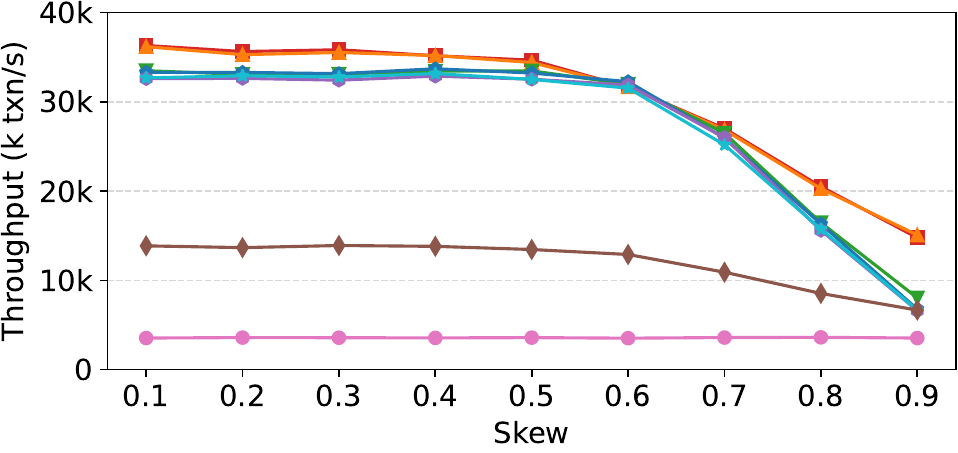}
	\caption{Overall Throughput on SmallBank}
	\label{OverallThroughputonSmallBank}
\end{figure}

As illustrated in Fig.~\ref{fig:OverallThroughputonYCSB}, BP yields prominent performance gains under severe data contention. Specifically, when the Zipfian skew reaches $0.9$, Lantern ($6,131$~txn/s) outperforms Lantern w/o BP ($2,339$~txn/s) by $2.6\times$.
Compared to baseline systems, Lantern maintains a substantial throughput lead:
\begin{itemize}[leftmargin=*]
	\item \textbf{VS. Aria:} Under extreme contention ($\text{skew} = 0.9$), Lantern outperforms Aria ($1,028$~txn/s) by $5.9\times$. 
	This performance gap stems from Aria's write-after-write (WAW) prohibition, which restricts each key to being written by at most one transaction within a batch to facilitate deterministic reordering. 
	Under high contention, this restriction increases transaction aborts, limiting Aria's throughput. 
	In contrast, Lantern allows multiple transactions to write the same key within a batch and retains the value written by the transaction with the highest index after Back-Propagation, thereby preserving the overwrite semantics of serial execution.
	\item \textbf{VS. ChainMaker:} Lantern consistently outperforms ChainMaker by at least $2.2\times$ across all contention levels. 
	This advantage stems from the fact that ChainMaker's nondeterministic protocol requires an additional replay stage to ensure cross-node consistency, thereby bottlenecking its overall throughput.
	In contrast, Lantern's deterministic execution enables all nodes to independently execute transactions in parallel and produce identical world states, eliminating the need for an additional replay stage.
\end{itemize}

On SmallBank (Fig.~\ref{OverallThroughputonSmallBank}), the throughput curves of Lantern and Lantern w/o BP nearly overlap, as do those of Lantern w/o CFBS and Lantern w/o \{BP, CFBS\}. This observation indicates that BP yields only marginal performance improvements under the RMW workload. Similarly, the throughput curves of Aria and Aria w/o DR nearly overlap, suggesting that deterministic reordering provides limited performance benefits under this workload.
Conversely, the performance gap between variants with CFBS (Lantern and Lantern w/o BP) and those without (Lantern w/o CFBS and Lantern w/o \{BP, CFBS\}) demonstrates the role of the CFBS mechanism in handling RMW workloads. 
Specifically, CFBS boosts throughput by up to $1.8\times$ under high contention ($\text{skew} = 0.9$, elevating throughput from $8,035$~txn/s for Lantern w/o CFBS to $14,775$~txn/s for Lantern).
Furthermore, compared to the fully unoptimized baseline (Lantern w/o \{BP, CFBS\}), Lantern's combined optimization achieves a $2.2\times$ speedup.

\subsection{Scalability}

To eliminate potential performance interference caused by non-uniform memory access (NUMA) in our dual-socket Intel Xeon architecture, we enforce a balanced resource allocation strategy for scalability evaluations across 4, 8, 16, and 32 CPU cores. 
Specifically, we evenly bind CPU cores across the two NUMA sockets using \texttt{taskset}, and enable memory page interleaving via \texttt{numactl --interleave=all}.
This setup ensures symmetric CPU scaling and uniform memory access latency across all core configurations.

Fig.~\ref{fig:ScalabilityonYCSB} presents the scalability results under the YCSB benchmark.
Under a uniform distribution (Fig.~\ref{fig:YCSBUniform}), data conflicts are extremely rare. Consequently, the BP mechanism yields marginal performance gains, causing the throughput curves of Lantern and Lantern w/o BP to virtually overlap. 
Under severe contention (Fig.~\ref{fig:YCSBZipfian}), Lantern consistently outperforms Lantern w/o BP across all core configurations. 

Fig.~\ref{fig:ScalabilityonSmallBank} illustrates the scalability results under SmallBank. 
When workload access is uniform (Fig.~\ref{fig:SmallBankUniform}), protocols integrated with the CFBS mechanism (Lantern and Lantern w/o BP) consistently outperform those without CFBS (Lantern w/o CFBS and Lantern w/o \{BP, CFBS\}).
Under severe contention (Fig.~\ref{fig:SmallBankZipfian}), Lantern and Lantern w/o BP outperform all other baselines by a wide margin.
When CFBS is enabled, the performance curves of Lantern and Lantern w/o BP virtually overlap, indicating that BP provides marginal benefit when conflicts are proactively eliminated during batch selection. 
However, when CFBS is disabled, BP takes effect by salvaging valid transactions from the dependency graph; specifically, Lantern w/o CFBS achieves $7,990$~txn/s at 32 cores, outperforming Lantern w/o \{BP, CFBS\} ($6,886$~txn/s) by $1.2\times$.

\begin{figure}[t] 
	\centering
	
	\includegraphics[width=\columnwidth]{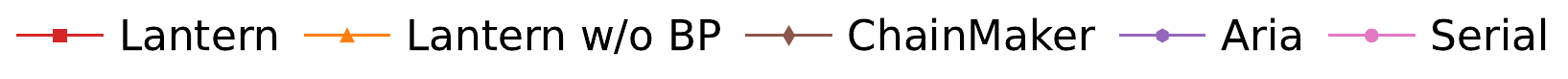}
	
	\vspace{1mm}
	
	\begin{subfigure}[b]{0.49\columnwidth} 
		\centering
		\includegraphics[width=\linewidth]{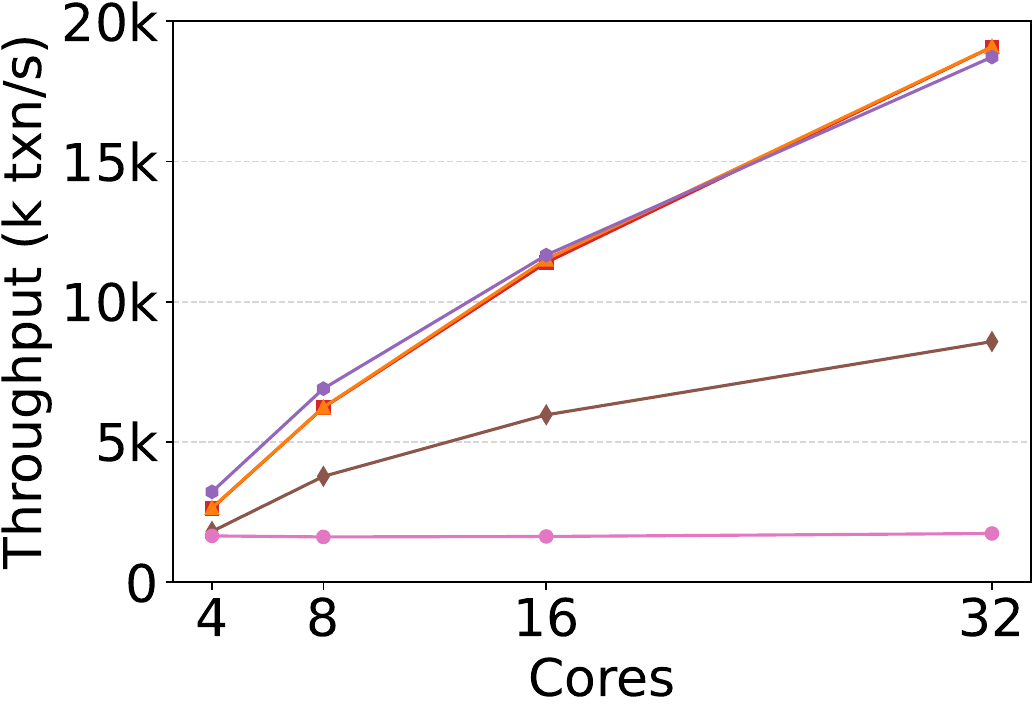}
		\caption{YCSB Uniform}
		\label{fig:YCSBUniform}
	\end{subfigure}
	\hfill
	\begin{subfigure}[b]{0.49\columnwidth}
		\centering
		\includegraphics[width=\linewidth]{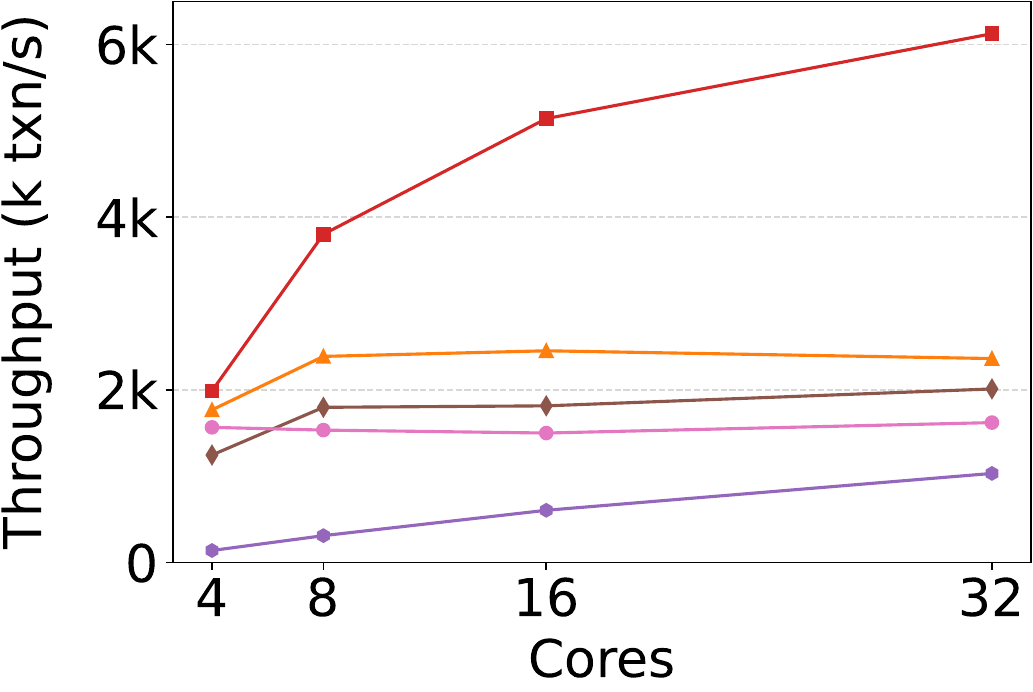}
		\caption{YCSB Zipf (Skew = 0.9)}
		\label{fig:YCSBZipfian}
	\end{subfigure}
	
	\caption{Scalability on YCSB}
	\label{fig:ScalabilityonYCSB}
\end{figure}

\begin{figure}[t] 
	\centering
	
	\includegraphics[width=\columnwidth]{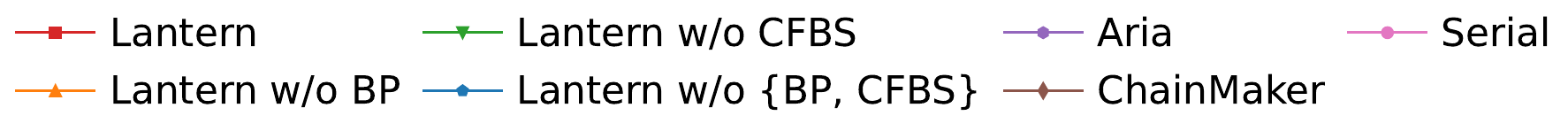}
	
	\vspace{1mm}
	
	\begin{subfigure}[b]{0.49\columnwidth} 
		\centering
		\includegraphics[width=\linewidth]{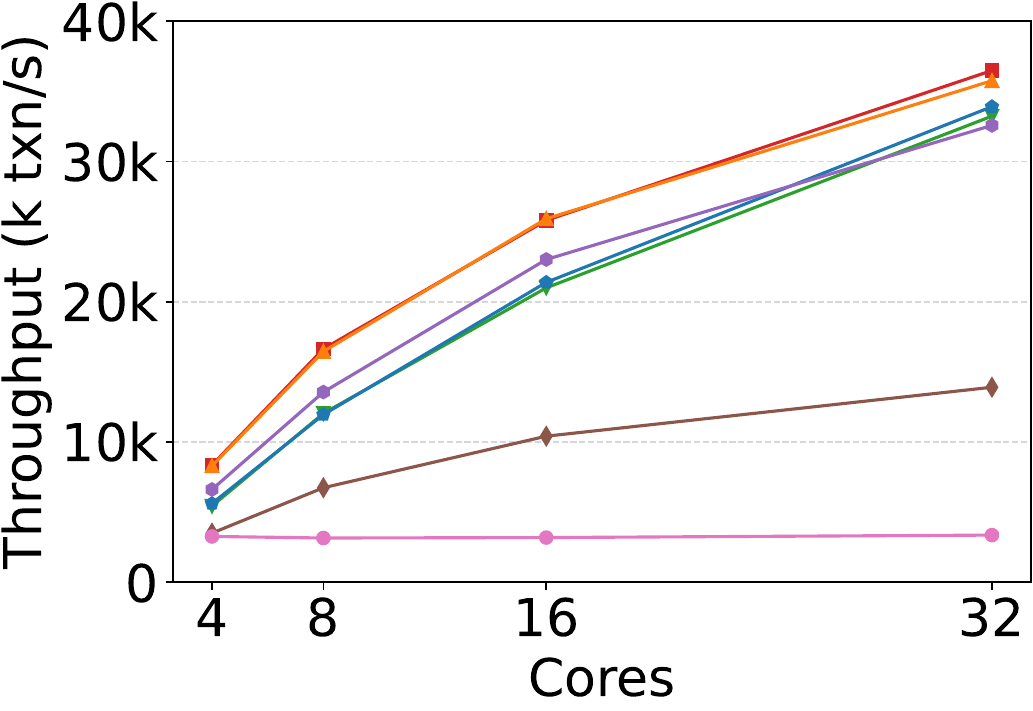}
		\caption{SmallBank Uniform}
		\label{fig:SmallBankUniform}
	\end{subfigure}
	\hfill
	\begin{subfigure}[b]{0.49\columnwidth}
		\centering
		\includegraphics[width=\linewidth]{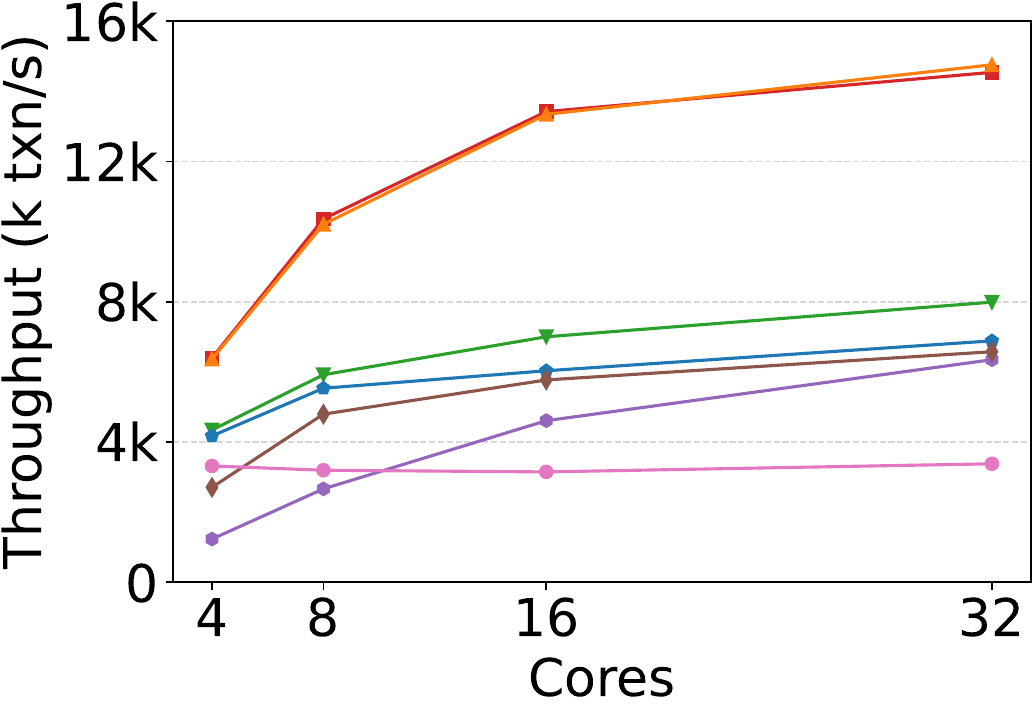}
		\caption{SmallBank Zipf (Skew = 0.9)}
		\label{fig:SmallBankZipfian}
	\end{subfigure}
	
	\caption{Scalability on SmallBank}
	\label{fig:ScalabilityonSmallBank}
\end{figure}

\subsection{Latency Breakdown}
\begin{figure}[t] 
	\centering
	
	\includegraphics[width=\columnwidth]{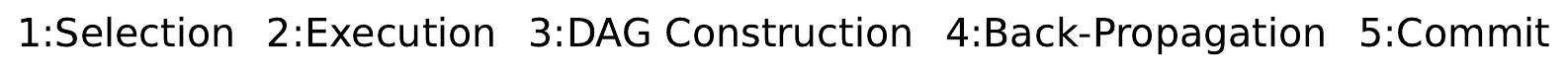}
	
	\vspace{1mm}
	
	\begin{subfigure}[b]{0.49\columnwidth} 
		\centering
		\includegraphics[width=\linewidth]{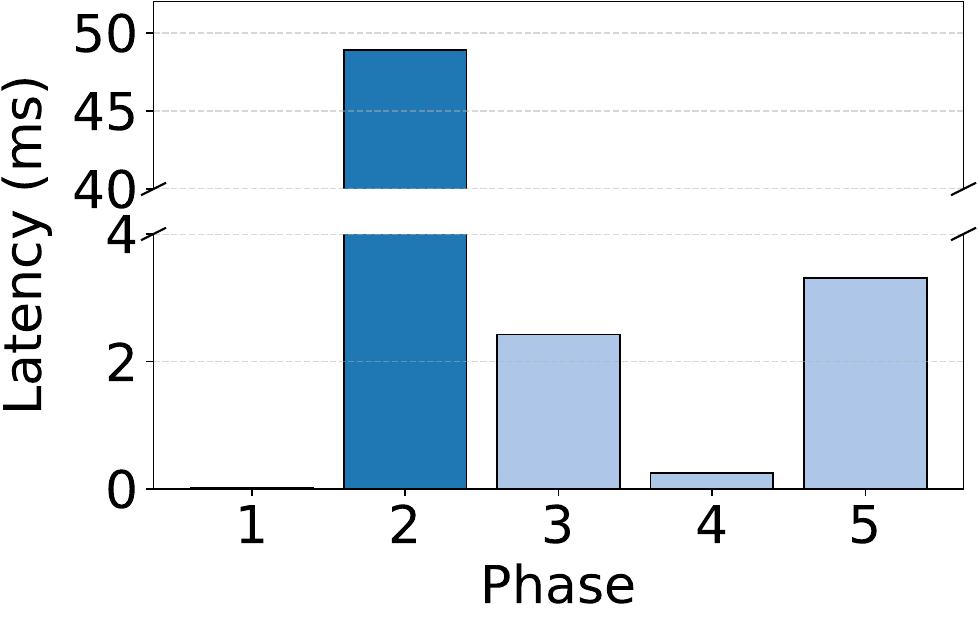}
		\caption{Low Contention (Skew = 0.1)}
		\label{fig:HighContention(Skew=0.1)}
	\end{subfigure}
	\hfill
	\begin{subfigure}[b]{0.49\columnwidth}
		\centering
		\includegraphics[width=\linewidth]{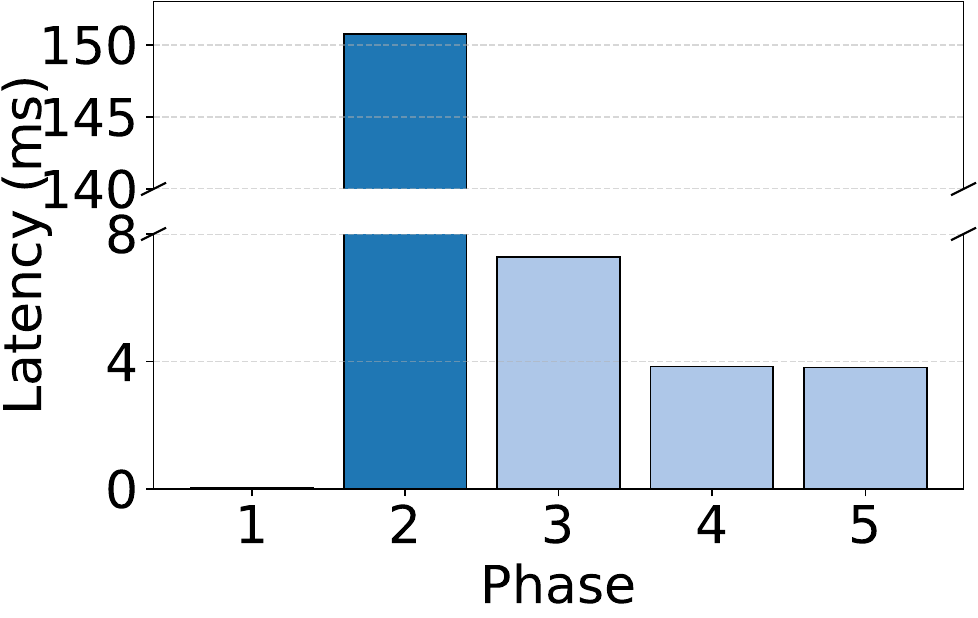}
		\caption{High Contention (Skew = 0.9)}
		\label{fig:HighContention(Skew=0.9)}
	\end{subfigure}
	
	\caption{Phase-wise Latency Breakdown On YCSB}
	\label{fig:Phase-wiseLatencyBreakdown}
\end{figure}

\begin{figure}[t] 
	\centering
	
	\includegraphics[width=\columnwidth]{figures/legend_YCSB_pb}
	
	\vspace{1mm}
	
	\begin{subfigure}[b]{0.49\columnwidth} 
		\centering
		\includegraphics[width=\linewidth]{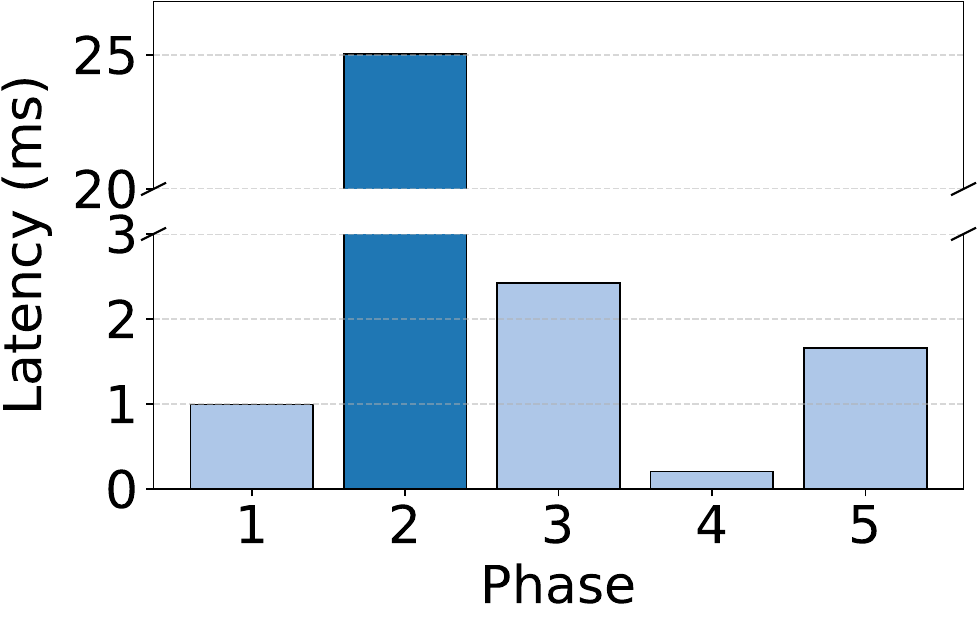}
		\caption{Low Contention (Skew = 0.1)}
		\label{Skew=0.1)}
	\end{subfigure}
	\hfill
	\begin{subfigure}[b]{0.49\columnwidth}
		\centering
		\includegraphics[width=\linewidth]{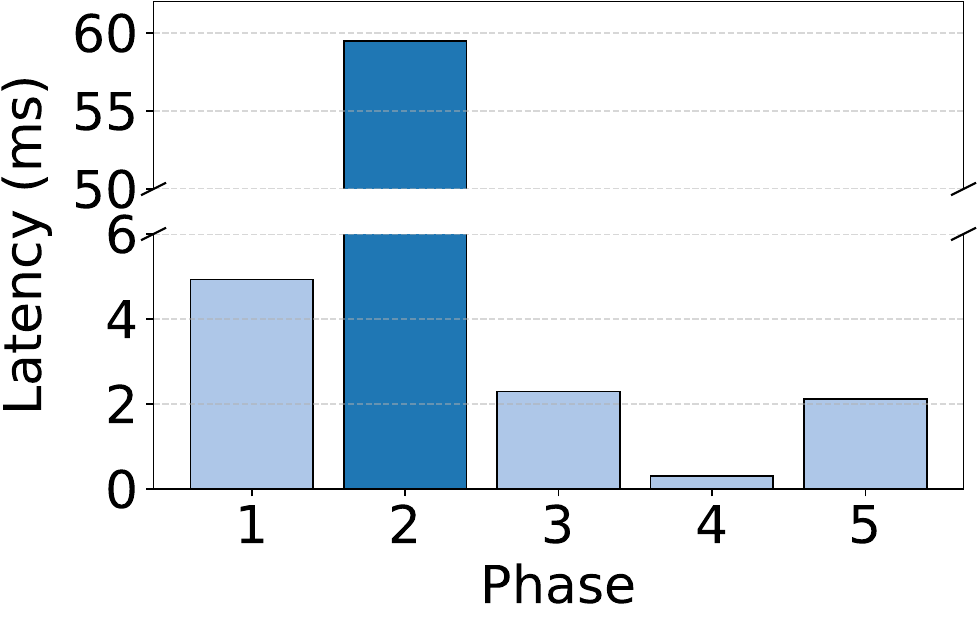}
		\caption{High Contention (Skew = 0.9)}
		\label{Skew=0.9)}
	\end{subfigure}
	
	\caption{Phase-wise Latency Breakdown On SmallBank}
	\label{Phase-wiseLatencyBreakdownOnSmallBank}
\end{figure}

In this section, we evaluate the phase-wise latency breakdown of Lantern under both low-contention ($\text{Zipfian skew} = 0.1$) and high-contention ($\text{Zipfian skew} = 0.9$) workloads. 
Fig.~\ref{fig:Phase-wiseLatencyBreakdown} illustrates the detailed latency distribution on YCSB. 
Under low contention (Fig.~\ref{fig:HighContention(Skew=0.1)}), transaction execution dominates the overall latency, with the Execution phase consuming $48.9$~ms, followed by the Commit phase at $3.3$~ms. 
The remaining scheduling operations—including Selection, DAG Construction, and Back-Propagation—account for only $4.9\%$ of the total latency, demonstrating Lantern’s minimal scheduling overhead.
Under high contention (Fig.~\ref{fig:HighContention(Skew=0.9)}), severe conflicts force aborted transactions into subsequent rounds for re-execution, increasing the Execution phase latency to $150.8$~ms.
On SmallBank (Fig.~\ref{Phase-wiseLatencyBreakdownOnSmallBank}), where the CFBS mechanism is enabled to proactively prevent hot-account conflicts, the Selection phase incurs $1.0$~ms and $4.9$~ms under skew $0.1$ and $0.9$, respectively.

\subsection{Impact of Batch Size $B_{\mathit{size}}$}
\label{Impact of Batch Size}
In this section, we evaluate the impact of the batch size $B_{\mathit{size}}$ on throughput. 
We compare five configurations where $B_{\mathit{size}}$ is fixed to $1\times$, $3\times$, $5\times$, $7\times$, and $9\times$ the CPU physical core count (denoted as Lantern-1 through Lantern-9). 
Note that on the SmallBank benchmark, the CFBS mechanism is disabled to prevent $B_{\mathit{size}}$ from becoming dynamic during processing.
Fig.~\ref{fig:ImpactofBatchSize} presents the results, revealing a clear performance trade-off across both workloads:
\begin{itemize}[leftmargin=*]
	\item \emph{Low-to-Moderate Contention ($\text{skew} \le 0.7$):} Larger batch sizes consistently achieve higher throughput. 
	Increasing $B_{\mathit{size}}$ from $1\times$ to $5\times$ yields substantial performance gains because a larger batch exposes greater intra-batch parallelism and amortizes scheduling overheads. 
	However, further increasing $B_{\mathit{size}}$ to $7\times$ and $9\times$ yields diminishing marginal returns as multi-core hardware execution capacity saturates.
	
	\item \emph{Severe Contention ($\text{skew} = 0.9$):} The throughput trend completely reverses. 
	Lantern-1 achieves the highest throughput, whereas throughput monotonically degrades as $B_{\mathit{size}}$ increases. 
	This degradation occurs because a large batch under high contention leads to low transaction commit ratios, and severe computational resource wastage from frequent aborts. 
	In contrast, a smaller batch size bounds the contention domain per execution round, maintaining a higher commit ratio and ensuring stable throughput.
\end{itemize}

Based on these findings, Lantern adopts $B_{\mathit{size}} = 5 \times \text{CPU cores}$ as its default configuration, balancing parallel throughput under low contention with robust handling of data conflicts.

\begin{figure}[t] 
	\centering
	
	\includegraphics[width=\columnwidth]{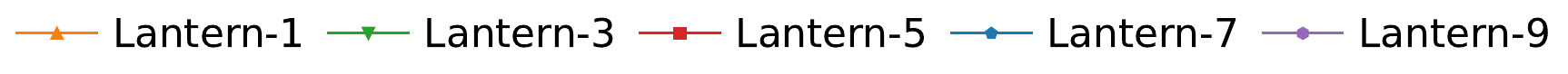}
	
	\vspace{1mm}
	
	\begin{subfigure}[b]{0.49\columnwidth} 
		\centering
		\includegraphics[width=\linewidth]{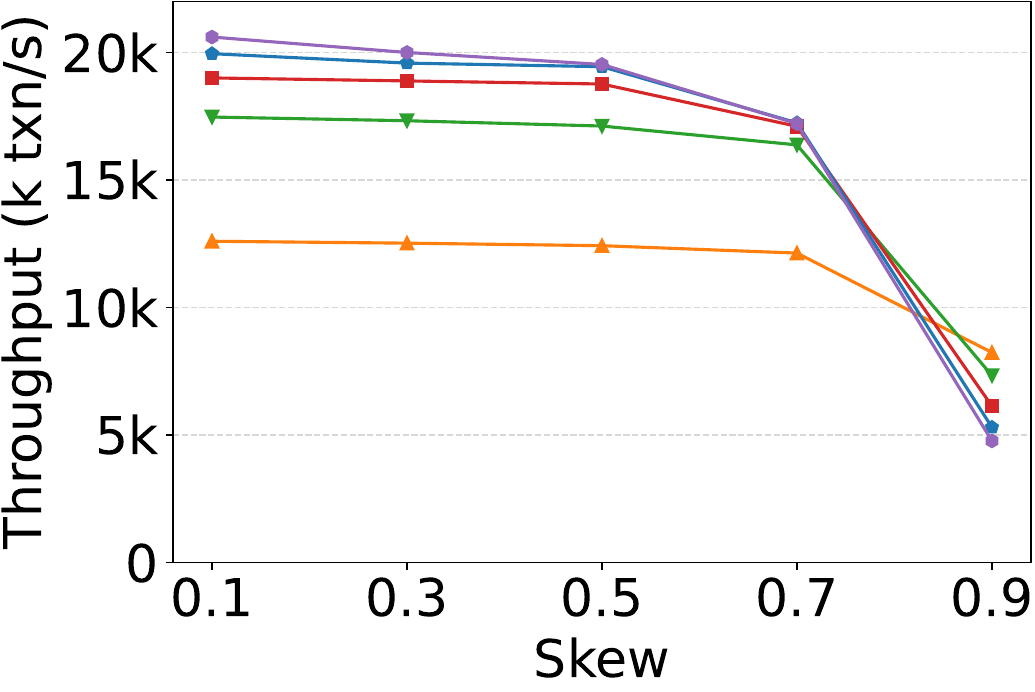}
		\caption{YCSB}
		\label{fig:YCSB}
	\end{subfigure}
	\hfill
	\begin{subfigure}[b]{0.49\columnwidth}
		\centering
		\includegraphics[width=\linewidth]{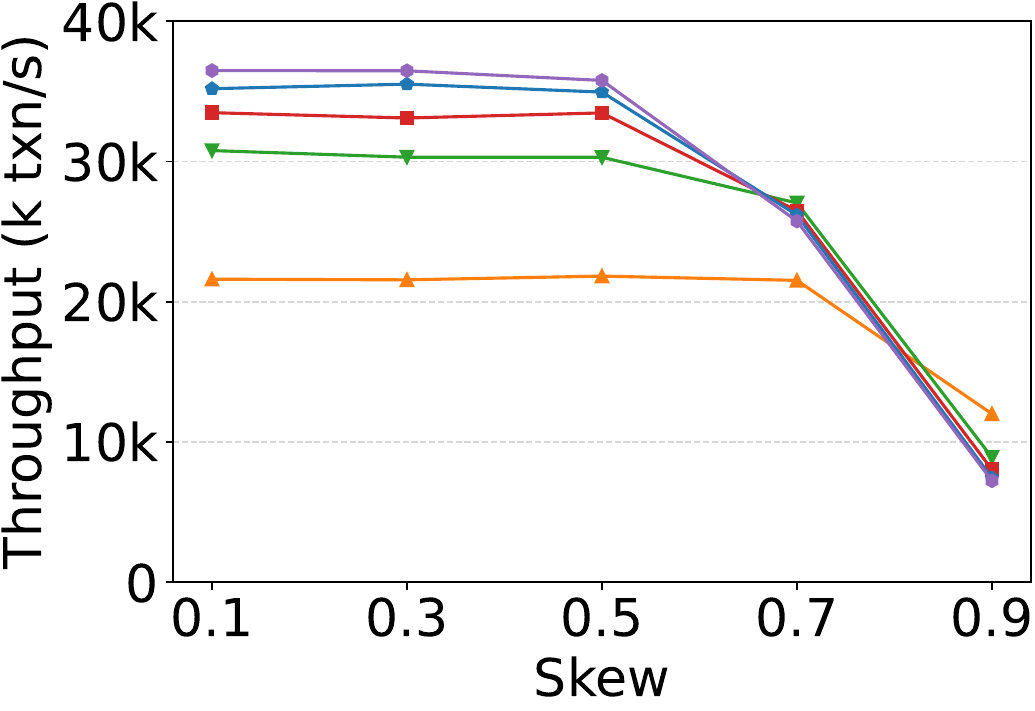}
		\caption{SmallBank}
		\label{fig:SmallBank}
	\end{subfigure}
	
	\caption{Impact of Batch Size $B_{\mathit{size}}$}
	\label{fig:ImpactofBatchSize}
\end{figure}

\section{Related Work}

Existing Deterministic concurrency control (DCC) protocols can be broadly classified into two categories based on whether they require prior knowledge of transaction read-write sets.
Early DCC protocols, including Calvin~\cite{Calvin}, BOHM~\cite{BOHM}, PWV~\cite{PWV}, Q-Store~\cite{QStore}, and Caracal~\cite{Caracal}, leverage prior transaction access information to facilitate deterministic scheduling and execution. 
These protocols differ in their scheduling and concurrency control mechanisms, but generally rely on access information available before transaction execution. 
Although transaction access information can often be declared or inferred in traditional databases, accurately determining read-write sets before execution remains challenging when transaction access patterns depend on runtime data. This challenge is particularly pronounced in blockchain systems, where smart contracts may contain data-dependent branches, loops, and cross-contract calls.
To remove this requirement, subsequent DCC protocols explore concurrency control mechanisms that determine transaction dependencies and resolve conflicts at runtime~\cite{LiTM,DOCC,aria,blockstm,Spectrum}. 
A representative protocol in this category is Aria~\cite{aria}, which adopts batch-based optimistic concurrency control, disallowing multiple transactions in a batch from writing the same data item (i.e., preventing write-after-write (WAW) conflicts) and employing deterministic reordering to commit transactions.
Lantern achieves higher performance than Aria through an overwrite-permissive, DAG-based strategy with deterministic Back-Propagation, which increases the number of committable transactions in each batch.
Moreover, unlike Aria's existential serializability guarantee, Lantern explicitly constructs an equivalent serializable order (Section~\ref{SerializabilityandDeterminism}).

\section{Conclusion}
In this paper, we presented Lantern, a deterministic concurrency control protocol that requires no prior read-write knowledge.
By combining an overwrite-permissive strategy on DAGs with a deterministic Back-Propagation mechanism, Lantern significantly expands the set of committable transactions per batch while maintaining cross-node state consistency. Additionally, its Conflict-Free Batch Selection mechanism effectively eliminates contention under RMW-intensive scenarios. 
We integrated Lantern into ChainMaker and evaluated it on YCSB and SmallBank benchmarks. 
The results demonstrate that Lantern achieves up to a $5.9\times$ throughput speedup over Aria and improves ChainMaker's execution-layer throughput by at least $2.2\times$.

\bibliographystyle{IEEEtran}
\bibliography{references}
\end{document}